\documentclass[12pt]{article}

\usepackage{fullpage}
\usepackage{amsmath}
\usepackage{amsthm}
\usepackage{amsfonts}
\usepackage{mathrsfs}

\newcommand{\naturals}{\mathbb{N}}

\newcommand{\rationals}{\mathbb{Q}}
\newcommand{\reals}{\mathbb{R}}

\newcommand{\Astruct}{\mathbb{A}}
\newcommand{\Bstruct}{\mathbb{B}}
\newcommand{\Cstruct}{\mathbb{C}}

\newcommand{\tp}{\mathrm{tp}}
\newcommand{\Dom}{\mathrm{Dom}}

\newcommand{\innprod}[2]{\langle #1,#2 \rangle}
\newcommand{\bits}{\mathbb{B}}

\newcommand{\vectors}{\mathrm{vec}}

\newcommand{\disjointunion}{\mathbin{\dot{\cup}}}
\newcommand{\Sphere}[2]{S(#1,#2)}
\newcommand{\infnorm}[1]{||#1||_{\infty}}

\newcommand{\Iprod}[2]{\langle{#1},{#2}\rangle}
\newcommand{\norminf}[1]{\lVert #1 \rVert_{\infty}}
\newcommand{\normone}[1]{\lVert #1 \rVert_{1}}
\newcommand{\normtwo}[1]{\lVert #1 \rVert_{2}}

\newcommand{\identitymatrix}{\mathrm{I}}
\newcommand{\allonesmatrix}{\mathrm{J}}

\newcommand{\POP}{\mathrm{POP}}
\newcommand{\Sol}{\mathrm{SOL}}
\newcommand{\SDP}{\mathrm{SDP}}
\newcommand{\LP}{\mathrm{LP}}
\newcommand{\opt}{\mathrm{opt}}
\newcommand{\ISO}{\mathrm{ISO}}
\newcommand{\CC}{\mathrm{CC}}
\newcommand{\CCC}{\mathrm{CC'}}

\newcommand{\SA}{\mathrm{SA}}
\newcommand{\PC}{\mathrm{PC}}
\newcommand{\SOS}{\mathrm{SOS}}

\newcommand{\Cinf}[1]{\mathrm{C}^{#1}_{\infty\omega}}
\newcommand{\FPC}{\mathrm{FPC}}
\newcommand{\FOC}{\mathrm{FOC}}

\newtheorem{theorem}{Theorem}[section]
\newtheorem{lemma}[theorem]{Lemma}
\newtheorem{fact}[theorem]{Fact}
\newtheorem{corollary}[theorem]{Corollary}
\newtheorem{claim}[theorem]{Claim}
\newtheorem{proposition}[theorem]{Proposition}

\numberwithin{equation}{section}

\begin{document}

\title{\bf Definable Ellipsoid Method, Sums-of-Squares Proofs, and the
  Graph Isomorphism Problem}

\author{Albert Atserias \\
Universitat Polit\`ecnica de Catalunya \\
 Barcelona, Catalonia, Spain
\and
Joanna Fijalkow \\
University of Bordeaux,
 CNRS, Bordeaux INP, LaBRI, UMR 5800 \\
  F-33400 Talence, France, and \\ Institute of Informatics, University of Warsaw \\
  Warsaw, Poland \\
}

\date{}

\maketitle

\begin{abstract}
The ellipsoid method is an algorithm that solves the (weak)
feasibility and linear optimization problems for convex sets by making
oracle calls to their (weak) separation problem. We observe that the
previously known method for showing that this reduction can be done in
fixed-point logic with counting (FPC) for linear and semidefinite
programs applies to any family of explicitly-bounded convex sets.
We
further
show that the exact feasibility problem for
semidefinite programs is expressible in the infinitary version of
FPC. As a corollary we get that, for the graph
isomorphism problem, the
Lasserre/Sums-of-Squares semidefinite programming hierarchy of
relaxations collapses to the Sherali-Adams linear programming
hierarchy, up to a small loss in the degree.
\end{abstract}

\section{Introduction} \label{sec:introduction}

Besides being the first algorithm to be discovered that could solve
linear programs (LPs) in polynomial time, the ellipsoid method has at
least two other 
features that make it an important tool for the
computer science theoretician.  The first is that 
it is
able to handle 
implicit LPs
given by exponentially many, or even infinitely many, linear
inequalities. These include some of the most fundamental problems of
combinatorial optimization and mathematical programming, such as the
weighted matching problem on general graphs, the submodular function
minimization problem, or approximately solving semidefinite programs.
The second important feature of the ellipsoid method is that, for LPs,
its running time is polynomial in the bit length of its input and 
is provably robust against issues of numerical instability (see,
e.g.,~\cite{GroetschelLovaszSchrijverBook}).

There is a third emerging 
feature of
the ellipsoid method that is of particular significance for the
logician and the descriptive complexity theorist. The starting point
is the important breakthrough result of Anderson, Dawar and Holm~\cite{Anderson:2015} who developed a method called \emph{folding}
to deal with symmetries in an LP.  They used this method 
to show that, for the special case of LPs, the ellipsoid method can be
implemented in fixed-point logic with counting (FPC), and hence in
polynomial time, but \emph{choicelessly}, i.e., in a way that the
symmetries of
the input are respected all along the computation, 
and in the output. As the main application of their result, they
proved that the class of graphs that have a perfect matching could be
defined in~FPC, thus solving one of the 
open problems
raised by Blass, Gurevich and Shelah in their work on Choiceless Polynomial
Time~\cite{Blass2002}. The method of folding was extended further
by Dawar and Wang 
to deal with explicitly-bounded and
full-dimensional semidefinite programs~(SDPs)~\cite{DawarW17}.

Our first contribution is the observation that the
method of folding
can be used to capture the power of the ellipsoid method in its
full strength. We observe that the fully general polynomial-time
reduction that solves the weak feasibility problem given a weak
separation oracle for an explicitly-bounded convex set can be
implemented, choicelessly, in FPC. As in the earlier works that employed
the folding method, our implementation
uses the reduction
algorithm as described in~\cite{GroetschelLovaszSchrijverBook} as a
black-box. The black-box is made into a choiceless procedure through a
sequence of runs of the algorithm along a refining sequence of
suitable quotients of the given convex set. It should be pointed out
that while all the main ideas for doing this were already implicit in
the earlier works by Anderson, Dawar and Holm, and by Dawar and Wang,
working out the details requires a certain degree of care. 
For example, when we started this work it was not clear whether the earlier
methods would be able to deal with separation oracles for families of
convex sets that are \emph{not} closed under the folding-quotient
operations.  We observe that such closure conditions, which happen to
hold for LPs and SDPs, are 
not required. The details of this can be found in
Section~\ref{sec:definableellipsoid}.

Aided by this new understanding, we develop three applications
of folding.

\subsection{The SDP exact feasibility problem}

The first application concerns the semidefinite programming exact
feasibility problem. A semidefinite set, also known as a
spectrahedron, is a subset of Euclidean space that is defined as the
intersection of the cone of positive semidefinite matrices with an
affine subspace. Thus, semidefinite sets are the feasible regions of
SDPs. The SDP exact feasibility problem asks, for an SDP given as
input, whether its feasible region is non-empty.
While the approximate and explicitly-bounded version of this problem
is solvable in polynomial-time by the ellipsoid method, the
computational complexity of \emph{exact} feasibility is a
well-known open problem in mathematical programming: it is decidable
in polynomial space, by reduction to the existential theory of the
reals, but its precise position in the complexity hierarchy is
unknown. 
It has been shown that the problem is at least as hard as PosSLP, the
positivity problem for integers represented as arithmetic
circuits~\cite{TarasovVyalyi2008}, and hence at least as hard as the
famous square-root sum problem, but the exact complexity of these two
problems is also largely unknown (see~\cite{Allender:2009}). 

Our result on the SDP exact feasibility problem is that, when its
input is represented suitably as a finite structure, it is definable
in the logic~$\Cinf{\omega}$, i.e., bounded-variable infinitary logic
with counting (see Section~\ref{sec:preliminaries} for definitions and
references for all logics appearing in this paper). In more recent
terminology, we say that the SDP exact feasibility problem has
\emph{bounded~counting~width}: there is a fixed bound~$k$ so that the
set of YES (and NO) instances of the problem is closed under
indistinguishability by formulas of~$k$-variable counting logic.
Let us briefly discuss the new idea that goes into proving this.

First we show that the~$\FPC$-definability of the ellipsoid method can
be combined with the techniques in~\cite{DawarW17} to give
an~FPC-formula~$\phi$ that solves the weak feasibility problem for
explicitly-bounded SDPs. This deviates from the result
in~\cite{DawarW17} in that it removes one of their two assumptions:
the full-dimensionality requirement is now dropped.
Then we show how to reduce the exact feasibility problem for
arbitrary SDPs to the weak feasibility problem for
explicitly-bounded~SDPs, and do so in bounded-variable infinitary
logic. What drives this reduction is the observation that an
arbitrary~SDP is feasible if, and only if, there exists a large
radius~$R{>}0$ such that, for every small tolerance~$\epsilon{>}0$,
an~$\epsilon$-perturbation of the constraints of the original~SDP
restricted to solutions of magnitude at most~$R$ is non-empty.
Verifying this last condition when an~$(R,\epsilon)$-pair is given in
the input can be done in~FPC through the formula~$\phi$: indeed, the
resulting~SDP is explicitly-bounded thanks to~$R$, and it is enough to
decide its weak feasibility thanks to~$\epsilon$.  Hence, the
reduction boils down to handling the~$\exists R{>}0\ \forall
\epsilon{>}0$ quantification in bounded-variable infinitary logic.  To
achieve this, the key observation is that the part of the input that
corresponds to an~$(R,\epsilon)$-pair is independent of the
original~SDP.  This allows us to construct a \emph{Booleanized}
version~$\phi_{R,\epsilon}$ of the FPC-formula~$\phi$ that works only
for the fixed~$(R,\epsilon)$-pair.  Finally, by replacing the~$\exists
R{>}0\ \forall \epsilon{>}0$ quantification by an infinite disjunction
and conjunction, respectively, we obtain the~$\Cinf{\omega}$
formula~$\bigvee_{R>0} \bigwedge_{\epsilon>0} \phi_{R,\epsilon}$.  We
analyze the exact form of~$\phi$ and show that it allows for the
operation of fixing~$R$ and~$\epsilon$ while retaining the same number of
variables. This is the subject of Section~\ref{sec:sdp}.

\subsection{The SOS proof-existence problem}

A Sums-of-Squares (SOS) proof that an $n$-variable polynomial
inequality~$p_0 \geq 0$ holds on a 
set
defined by the polynomial constraints~$p_1 \geq 0,\ldots,p_m \geq 0$
is an
identity of the form~$\sum_{j=1}^m p_j s_j + s_0 =
p_0$, where each polynomial~$s_j$
is a sum of squares
of polynomials. The sums-of-squares methodology for solving polynomial
optimization problems advocated by Lasserre~\cite{Lasserre2001} and
Parrilo~\cite{Parrilo2000} motivates the question of computing such
proofs, when they exist. It is well-known that, in many settings,
including in the case 
of polynomial inequalities over Boolean variables,
the search-space of SOS proofs
with polynomials bounded by degree~$d$ can be formulated as the
feasible region of an~SDP with~$m\cdot n^{O(d)}$ variables and
constraints, and small coefficients. This leads to a computational
approach to finding low-degree SOS proofs by reduction to the~SDP
exact feasibility problem. It should be noted that a naive application
of this method does \emph{not}
in general
yield algorithms that are polynomial
in~$n$ and~$m$ even for~$d = O(1)$ due to results
in~\cite{DBLP:conf/innovations/ODonnell17} proving that SOS proofs
suffer from blow-up phenomena
in the coefficients of their
polynomials.
Remarkably, it was later shown
in~\cite{DBLP:conf/icalp/RaghavendraW17} that in certain special cases
of the problem the blow-up 
phenomena do not appear.

We recall the~SDP that describes the search-space of low-degree SOS
proofs and note that its representation
as a finite relational structure
is computable in the logic~FPC from a natural representation of
the input polynomials~$p_0,p_1,\ldots,p_m$.
Together with the definability of the SDP exact feasibility
problem, this
implies that the SOS proof system over Boolean
variables is \emph{weakly degree-automatable} in the
logic~$\Cinf{\omega}$
in the following sense:
there is a constant~$c$ such that, for each
degree~$d$, there is a formula~$\phi_d$ of the logic~$\Cinf{cd}$ that
tells whether a given polynomial inequality~$p_0 \geq 0$ has a
degree-$d$ SOS proof from a given system~$p_1\geq 0,\ldots,p_m\geq 0$
of polynomial constraints over Boolean variables. The qualification
\emph{weakly} in degree-automatable distinguishes the problem from its
search version in which an actual degree-$d$ proof is sought. For
refutations, 
where~$p_0$ is 
the constant $-1$ polynomial,
we note that the proof-existence problem can also be
reduced, in FPC, to the 
weak feasibility problem for arbitrary SDPs (not necessarily 
explicitly-bounded).
While less demanding, this weaker form of the
problem is not known to be solvable in polynomial time, let alone
FPC definable.
All this
can be found in Section~\ref{sec:sos}.

While interesting in its own right for its potential applications to
proof complexity lower bounds along the lines
of~\cite{DBLP:journals/lmcs/GradelGPP19}, we think of
the weak degree-automatability result for SOS proofs as the required
tool to develop the third and main application.

\subsection{Hierarchies for the graph isomorphism problem}

A variety of mathematical programming relaxations of the graph
isomorphism problem have been proposed in the literature: the
fractional isomorphism relaxation of Tinhofer~\cite{Tinhofer1986}, its
strengthening via the Sherali-Adams hierarchy of~LP
relaxations~\cite{Atseriasdoi10,MALKIN201473}, its further
strengthening via the Lasserre hierarchy of~SDP
relaxations~\cite{O'Donnell:2014}, its relaxation via Groebner basis
computations~\cite{GroheBerkholz15}, and a few others.
While it is known that no fixed level of any of these hierarchies
of~LP,~SDP or Groebner-based relaxations solves the graph isomorphism
problem~\cite{Atseriasdoi10,MALKIN201473,O'Donnell:2014,GroheBerkholz15},
their relative strength was not fully understood before our work.
Since~SDP is a proper generalization of~LP, one may be tempted to
guess that the Lasserre SDP hierarchy could perhaps distinguish more
graphs than its Sherali-Adams~LP sibling.
Our main contribution
is to prove that this is not the case: for
the graph isomorphism problem, the strength of the Lasserre hierarchy
collapses to that of the Sherali-Adams hierarchy.

Concretely, we prove in Section~\ref{sec:isomorphism} that there
exists a constant~$c$ such that if two graphs are distinguishable at
level~$d$ of the Lasserre hierarchy, then they are also
distinguishable at level~$cd$ of the Sherali-Adams hierarchy.  The
constant~$c$ loss comes from the number of variables
that are needed to express
the~SDP exact feasibility
problem in bounded-variable infinitary logic with counting.  
This collapse may sound surprising because it implies that, for
distinguishing graphs, the spectral methods that underlie the Lasserre
hierarchy are already available in low levels of the Sherali-Adams
hierarchy. However, it 
agrees nicely
with the
fact that indistinguishability by~3-variable counting logic captures
graph spectra~\cite{DawarSeveriniZapata} and the correspondence
between~$k$-variable counting logic and level~$k$ of the Sherali-Adams
hierarchy from~\cite{Atseriasdoi10}.  It also aligns well with the
results in~\cite{ODonnellSchramm2021} where it is shown that certain
spectral methods for approximating the number of constraints that can
be satisfied in a constraint satisfaction problem can be implemented
directly in the Sherali-Adams hierarchy.

To get the collapse result we
consider the
standard~$0$-$1$ quadratic programming formulation~$P(G,H)$ of the
graph isomorphism problem for graphs~$G$ and~$H$ on disjoint sets of
vertices.
Assuming
that the level-$d$
Lasserre relaxation of~$P(G,H)$ distinguishes~$G$ and~$H$,
our new insights 
on the expressibility of the~SDP exact
feasibility problem 
imply
that~$G$ and~$H$ can be distinguished by 
a~$\Cinf{cd}$-sentence,
where~$c$ is a constant independent of~$d$.  Hence, by the main result
in~\cite{Atseriasdoi10} relating the levels of the Sherali-Adams
hierarchy with indistinguishability in bounded-variable counting
logic,
the graphs~$G$ and~$H$ can be distinguished by level-$cd$
Sherali-Adams relaxation, thus proving the collapse. It should be
noted that, remarkably, this holds for any two graphs and
\emph{any}~$d$, even if~$d = d(n)$ is an arbitrary function of the
number~$n$ of vertices of~$G$ and~$H$. The details of this can be
found in Section~\ref{sec:isomorphism}.

When stated in the language of proofs, the collapse has another
interesting consequence. By combining the results
in~\cite{Atseriasdoi10} and~\cite{GroheBerkholz15}, it was already
known that if there is a degree-$d$ Sherali-Adams (SA)
proof that~$G$ and~$H$ are not
isomorphic, then there is also a degree-$d$ monomial~Polynomial
Calculus (mon-PC) proof over the reals, hence also a degree-$d$
Polynomial Calculus (PC) proof over the reals, which implies that
there is a degree-$2d$ SOS proof
by~\cite{DBLP:journals/eccc/Berkholz17}. Thus, for the graph
isomorphism problem, our collapse result completes a full cycle of
simulations~$\text{SOS}_{d} \to \text{SA}_{cd} \to \text{mon-PC}_{cd}
\to \text{PC}_{cd} \to \text{SOS}_{2cd}$ to show that all these proof
systems are equally powerful up to a~$2c$-factor loss in the degree.
It also confirms the belief expressed in~\cite{GroheBerkholz15} that
the gap between PC and monomial PC is not large (a result obtained
independently in~\cite{DBLP:journals/lmcs/GradelGPP19}).

It is remarkable that we
proved these statements
about the relative strength of proof systems and hierarchies through
an excursion into the descriptive complexity of the ellipsoid method,
the SDP exact feasibility problem, and bounded-variable
infinitary~logics.  However, it should be noted that our proof is
indirect as it relies on the correspondence between~$k$-variable
counting logic and the~$k$-th level Sherali-Adams hierarchy
from~\cite{Atseriasdoi10}.
The question whether the collapses can be
shown to hold \emph{directly} by strengthening~LP-solutions to~SDP-ones
for the primals, or by relaxing~SDP-solutions to~LP-ones for the
duals, remains an interesting~one.

\section{Preliminaries} \label{sec:preliminaries}

We use~$[n]$ to denote the set~$\{1,\ldots,n\}$.

\paragraph{Vectors and matrices}

If~$I$ is a non-empty index set, then an~$I$-vector is an element
of~$\reals^I$. The components of~$u \in \reals^I$ are written~$u(i)$
or~$u_i$, for~$i \in I$. We identify~$\reals^n$
with~$\reals^{[n]}$. For~$I$-vectors~$u$ and~$v$, the \emph{inner
  product} of~$u$ and~$v$ is~$\Iprod{u}{v} = \sum_{i \in I} u_iv_i$.
We write~$\normone{u} = \sum_{i \in I} |u_i|$ for
the~$L_1$-norm,~$\normtwo{u} = \sqrt{\Iprod{u}{u}}$ for
the~$L_2$-norm, and~$\norminf{u} = \max\{|u_i| : i \in I\}$ for
the~$L_{\infty}$-norm.  For~$K \subseteq \reals^I$ and~$\delta > 0$,
we define the \emph{$\delta$-ball} around~$K$ by~$\Sphere{K}{\delta}
:= \{ x \in \reals^I : \normtwo{x-y} \leq \delta \text{ for some } y
\in K\}$.  For~$K=\{x\}$, we set~$\Sphere{x}{\delta} :=
\Sphere{\{x\}}{\delta}$.  We define also~$\Sphere{K}{-\delta} := \{ x
\in \reals^I : \Sphere{x}{\delta} \subseteq K \}$.
When we refer to the \emph{volume} of a subset $K$ of Euclidean 
space $\reals^I$, we assume that $K$ is Lebesgue measurable and that the 
volume is defined as its Lebesgue measure (see, e.g.,~\cite{10.5555/26851}).
In particular, the volume of a $1$-ball in the $n$-dimensional real vector
space is $V_n = \pi^{n /2} / \Gamma(n/2+1)$,
where $\Gamma$ is the gamma function, i.e., the standard
continuous extension of the factorial function.

If~$I$ and~$J$ are two non-empty index sets, then an~$I \times
J$-matrix is simply an~$I \times J$-vector; i.e., an element
of~$\reals^{I \times J}$. Accordingly, the components of~$X \in
\reals^{I \times J}$ are written~$X(i,j)$, or~$X_{i,j}$, or~$X_{ij}$.
The~$L_1$-,~$L_2$- and~$L_\infty$-norms of a matrix~$X \in \reals^{I
  \times J}$ are defined as the respective norms of~$X$ seen as an~$I
\times J$-vector, and the inner product of the matrices~$X,Y \in
\reals^{I \times J}$ is~$\Iprod{X}{Y} = \sum_{i \in I}\sum_{j \in J}
X_{ij}Y_{ij}$.
Matrix product is written by concatenation.
A square matrix~$X \in \reals^{I \times I}$ is \emph{positive definite},
denoted~$X \succ 0$, if it is symmetric and satisfies~$z^T X z > 0$,
for every non-zero~$z \in \reals^I$. If it is symmetric but satisfies
the weaker condition that~$z^T X z \geq 0$, for every~$z \in
\reals^I$, then it is \emph{positive semidefinite}, which we denote by~$X
\succeq 0$. Equivalently,~$X$ is positive semidefinite if and only
if~$X = Y^T Y$ for some matrix~$Y \in \reals^{J \times I}$ if and only
if all its eigenvalues are non-negative.  By~$\identitymatrix$ we
denote the square identity matrix of appropriate dimensions,
i.e.,~$\identitymatrix_{ij} = 1$ if~$i = j$ and~$\identitymatrix_{ij}
= 0$ if~$i \not= j$. By~$\allonesmatrix$ we denote the square all-ones
matrix of appropriate dimensions, i.e.,~$\allonesmatrix_{ij} = 1$ for
all~$i$ and~$j$. For~$\identitymatrix$ and~$\allonesmatrix$ we omit
the reference to the index set in the notation (particularly so if the
index set is called~$I$ or~$J$, for obvious reasons).

\paragraph{Vocabularies, structures and logics}

A many-sorted (relational) vocabulary~$L$ is a set of sort
symbols~$D_1,\ldots,D_s$ together with a set of relation
symbols~$R_1,\ldots,R_m$.  Each relation symbol~$R$ in the list has an
associated \emph{type} of the form~$D_{i_1} \times \cdots \times
D_{i_r}$, where~$r \geq 0$ is the \emph{arity} of the symbol,
and~$i_1,\ldots,i_r \in [s]$ are not necessarily distinct. A
structure~$\Astruct$ of vocabulary~$L$, or an~$L$-structure, is given
by~$s$ disjoint sets~$D_1,\ldots,D_s$ called \emph{domains}, one for
each sort symbol~$D_i \in L$, and one relation~$R \subseteq D_{i_1}
\times \cdots \times D_{i_r}$ for each relation symbol~$R \in L$ of
type~$D_{i_1} \times \cdots \times D_{i_r}$. We use~$D(\Astruct)$
or~$D$ to denote the domain associated to the sort symbol~$D$,
and~$R(\Astruct)$ or~$R$ to denote the relation associated to the
relation symbol~$R$. In practice, the overloading of the notation
should never be an issue.  The domain of a sort symbol is also called
a \emph{sort}.  If~$\Astruct$ is an~$L$-structure and~$L'$ is a
many-sorted vocabulary obtained from~$L$ by removing some sort and
relation symbols, then an \emph{$L'$-reduct} of~$\Astruct$,
denoted~$L'(\Astruct)$, is the~$L'$-structure obtained from~$\Astruct$
by omitting the domains and relations associated to the sort and
relation symbols which are not present in~$L'$.

A logic for a many-sorted vocabulary~$L$ has an underlying set of
\emph{individual variables} for each different sort in~$L$. When
interpreted on an~$L$-structure, the variables are supposed to range
over the domain of its sort; i.e., the variables are typed. Besides
the equalities~$x = y$ between variables of the same type, the
atomic~$L$-formulas are the formulas of the form~$R(x_1,\ldots,x_r)$,
where~$R$ is a relation symbol of arity~$r$ and~$x_1,\ldots,x_r$ are
variables of types that match the type of~$R$. The formulas of
first-order logic over~$L$ are built from the atomic formulas by
negations, disjunctions, conjunctions, and existential and universal
quantification of individual variables. For detailed background
on first-order logic
see, e.g.,~\cite{DBLP:books/daglib/0080659}.

The syntax of First-Order Logic with Counting~$\FOC$ is defined by
adjoining one more sort~$N$ to the underlying vocabulary, adding one
binary relation symbol~$\leq$ of type~$N \times N$ and two ternary
relation symbols~$+$ and~$\times$ of types~$N \times N \times N$, as
well as extending the syntax to allow quantification of the
form~$\exists^{\geq y} x (\varphi)$, where~$\varphi$ is a formula,~$x$
is a variable of any type and~$y$ is a variable of type~$N$. In the
semantics of~$\FOC$, each~$L$-structure~$\Astruct$ is expanded to
an~$L \cup \{N,\leq,+,\times\}$-structure with~$N(\Astruct) =
\{0,\ldots,n\}$, where~$n = \max\{|D_i(\Astruct)| : i = 1,\ldots,s
\}$, and~$\leq$,~$+$, and~$\times$ are interpreted by the standard
arithmetic relations on~$\{0,\ldots,n\}$. The meaning
of~$\exists^{\geq y} x (\varphi)$, for a concrete assignment~$y
\mapsto i \in \{0,\ldots,n\}$, is that there exist at least~$i$ many
witnesses~$a$ for the variable~$x$ within its sort such that the
assignment~$x \mapsto a$ satisfies the formula~$\varphi$.  Numbers up
to~$n^c$, where~$c > 1$ is an integer, are represented by~$c$-tuples
of numbers in~$\{0,\ldots,n-1\}$. The arithmetic relations on such
numbers, and the quantifiers counting up to such numbers, are both
definable in~$\FPC$, the logic that we introduce next.

The syntax of Fixed-Point Logic with Counting~$\FPC$ extends the
syntax of~$\FOC$ by allowing the formation of \emph{inflationary
  fixed-point} formulas~$\mathrm{ifp}_{x,X} \varphi(x,X)$.  On a
structure~$\Astruct$ of the appropriate vocabulary, such formulas are
interpreted as defining the least fixed-point of the monotone
operator~$A \mapsto A \cup \{ a \in D_{i_1} \times \cdots \times
D_{i_r} : \Astruct \models \varphi(a,A) \}$, where~$D_{i_1} \times
\cdots \times D_{i_r}$ is the type of the relation symbol~$X$
in~$\varphi(x,X)$.  

The syntax of Infinitary Logic with
Counting~$\Cinf{}$ extends the syntax of first-order logic by allowing
quantifiers of the form~$\exists^{\geq i} x(\varphi)$ which say that
there are at least~$i$ many witnesses for the variable~$x$, where~$i$
is a (concrete) natural number, as well as infinite disjunctions and
conjunctions; i.e., formulas of the form~$\bigvee_{i \in I} \phi_i$
and~$\bigwedge_{i \in I} \phi_i$ where~$I$ is a possibly infinite
index set, and~$\{ \phi_i : i \in I \}$ is an indexed set of formulas.
The fragment of~$\Cinf{}$ with~$k$ variables, denoted~$\Cinf{k}$, is
the set of formulas that use at most~$k$ variables of any type. In the
formulas of~$\Cinf{k}$ the variables can be reused and hence there is
no finite bound on the quantification depth of the formulas.
We write~$\Cinf{\omega}$ for the union of the~$\Cinf{k}$ over all
natural numbers~$k$. It is well-known that for every natural
number~$k$, every many-sorted vocabulary~$L$, and
every~$L$-formula~$\varphi$ of~$\FPC$ that uses~$k$ variables, there
exists an~$L$-formula~$\psi$ of~$\Cinf{2k}$ such that~$\varphi$
and~$\psi$ define the same relations over all
finite~$L$-structures. 
While all the published proofs that we are
aware of give the statement for single-sorted vocabularies,
it is clear that the case of
many-sorted vocabularies is analogous.
For the proof and more on~$\FPC$ and~$\Cinf{\omega}$, we
refer to~\cite{OttoBook}.

\paragraph{Interpretations and reductions}

Let~$L$ and~$K$ be two many-sorted vocabularies, and let~$\Theta$ be a
class of~$K$-formulas. A~$\Theta$-interpretation of~$L$ in~$K$ is
given by: two~$\Theta$-formulas~$\delta_D(x)$ and~$\epsilon_D(x,y)$
for each sort symbol~$D$ of~$L$, and
one~$\Theta$-formula~$\psi_R(x_1,\ldots,x_r)$ for each relation
symbol~$R \in L$ of arity~$r$. In all these formulas, the
displayed~$x$'s and~$y$'s are tuples of distinct variables of the same
length~$m$, called the arity of the interpretation. We say that the
interpretation takes a~$K$-structure~$\Astruct$ as input and produces
an~$L$-structure~$\Bstruct$ as output if for each sort symbol~$D$
in~$L$ there exists a surjective partial map~$f_D : A^m \rightarrow
D(\Bstruct)$, where~$A$ is the domain of~$\Astruct$, such
that~$f_D^{-1}(D(\Bstruct)) = \{ a \in A^m : \Astruct \models
\delta_D(a) \}$,~$f_D^{-1}(\{(b,b) : b \in D(\Bstruct)\}) = \{ (a,b)
\in (A^{m})^2 : \Astruct \models \epsilon_D(a,b) \}$,
and~$f_R^{-1}(R(\Bstruct)) = \{ (a_1,\ldots,a_r) \in (A^m)^r :
\Astruct \models \psi_R(a_1,\ldots,a_r) \}$ where~$f_R = f_{D_1}
\times \ldots \times f_{D_r}$ and~$D_1 \times \cdots \times D_r$ is
the type of~$R$.
The composition of two interpretations, one of~$L$ in~$K$, and another
one of~$K$ in~$J$, is an interpretation of~$L$ in~$J$ defined in the
obvious way.  Similarly, the composition of an interpretation of~$L$
in~$K$ with an~$L$-formula is a~$K$-formula defined in the obvious
way. In all these compositions, the number of variables in the
resulting formulas \emph{multiply}. For example, the composition of
a~$\Cinf{k}$-interpretation with a~$\Cinf{\ell}$-formula is
a~$\Cinf{k\ell}$-formula. A reduction from a computational problem to
another is a pair of maps~$f$ and~$g$, where~$f$ takes an input~$x$
for the first problem and produces an input~$y = f(x)$ for the second
problem, and~$g$ takes~$x$ and a solution~$y'$ for~$y$ in the second
problem and produces a solution~$x' = g(x,y')$ for~$x$ in the first
problem.  The reduction is called a~$\Theta$-reduction if the maps can
be produced by~$\Theta$-interpretations when their inputs are
represented as structures of appropriate vocabularies.
For more on interpretations and logical reductions
see, e.g.,~\cite{DBLP:books/daglib/0082516}.

\paragraph{Numbers, vectors and matrices as structures}

Since we are interested in definability in logics, we represent
mathematical objects which serve as inputs and outputs of
algorithms as finite relational structures. The details of the chosen
representation are not essential, but we provide them for concreteness.

A natural number~$n \in \naturals$ is
represented by a structure, with a domain~$\{0,\ldots,N-1\}$ of
\emph{bit positions} where~$N \geq \lfloor{\log_2(n+1)}\rfloor$, of a
vocabulary~$L_{\naturals}$ that contains a binary relation
symbol~$\leq$ for the natural \emph{linear order} on the bit
positions, and a unary relation symbol~$P$ for the \emph{actual bits},
i.e., the bit positions~$i$ that carry a~$1$-bit in the unique binary
representation of~$n$ of length~$N$. Single bits~$b \in \{0,1\}$ are
represented as natural numbers with at least one bit position. Thus
the vocabulary~$L_{\bits}$ for representing single bits is really the
same as~$L_{\naturals}$, but we still give it a separate name.
A rational~$q = (-1)^b n/d$, where~$b \in \{0,1\}$ and~$n,d \in
\naturals$, is represented by a structure with
domain~$\{0,\ldots,N-1\}$ of bit positions, where~$N$ is large enough
to encode both the numerator~$n$ and the denominator~$d$ in
binary. The vocabulary~$L_{\rationals}$ of this structure has one
binary relation symbol~$\leq$ for the natural linear order on the bit
positions, and three unary relation symbols~$P_s$,~$P_n$ and~$P_d$
that are used to encode the sign and the bits of the numerator and the
denominator of~$q$. We use zero denominator to represent~$\pm \infty$.

An~$I$-vector~$u \in \rationals^I$ is represented by a two-sorted
structure, where the first sort~$\bar{I}$ is the index set~$I$ and the
second sort~$\bar{B}$ is a domain~$\{0,\ldots,N-1\}$ of bit positions,
where~$N$ is large enough to encode all the numerators and
denominators in the entries of~$u$ in binary.  The
vocabulary~$L_{\vectors}$ of this structure has one unary relation
symbol~$I$ for~$\bar{I}$, one binary relation symbol~$\leq$ for the
natural linear order on~$\bar{B}$, and three binary relation
symbols~$P_s$,~$P_n$ and~$P_d$, each of type~$\bar{I} \times \bar{B}$,
that are used to encode the entries of~$u$ in the expected
way:~$P_s(i,0)$ if and only if~$u(i)$ is positive,~$P_n(i,j)$ if and
only if the~$j$-th bit of the numerator of~$u(i)$ is~$1$,
and~$P_d(i,j)$ if and only if the~$j$-th bit of the denominator
of~$u(i)$ is~$1$.

More generally, if~$I_1,\ldots,I_d$ denote index sets that are not
necessarily pairwise distinct, then the corresponding tensors~$u \in
\rationals^{I_1 \times \cdots \times I_d}$ are represented by
many-sorted structures, with one sort~$\bar{I}$ for each index set~$I$
for as many different index sets as there are in the
list~$I_1,\ldots,I_d$, plus one sort~$\bar{B}$ for the bit
positions. The vocabulary~$L_{\vectors,d}$ of these structures has one
unary relation symbol~$I$ for each index sort~$\bar{I}$, one binary
relation symbol~$\leq$ for the natural linear order on the bit
positions~$\bar{B}$, and three~$d+1$-ary relation symbols~$P_s$,~$P_n$
and~$P_d$, each of type~$\bar{I}_1 \times \cdots \times \bar{I}_d
\times \bar{B}$, for encoding the signs and the bits of the numerators
and the denominators of the entries of the tensor.  Matrices~$A \in
\rationals^{I \times J}$ and square matrices~$A \in \rationals^{I
  \times I}$ are special cases of these, and so are indexed sets of
vectors~$\{ u_i : i \in K \} \subseteq \rationals^I$ and indexed sets
of matrices~$\{ A_i : i \in K \} \subseteq \rationals^{I \times J}$.

\section{The Definable Ellipsoid Method} \label{sec:definableellipsoid}

In this section we show that the ellipsoid method can be implemented
in~$\FPC$ for any family of explicitly-bounded convex sets. We begin
by defining the problems involved.

\subsection{Geometric problems and the ellipsoid method}

Let~$\mathscr{C}$ be a class of convex sets, each of the form~$K
\subseteq \mathbb{R}^I$ for some non-empty index set~$I$. 
We will consider elements of~$\mathscr{C}$ as inputs
of computational problems, and therefore the
class~$\mathscr{C}$ comes with an associated encoding scheme.
Most usual encoding schemes encode instances of a problem as
finite binary strings.
In our case, since we want to refer to definability in a logic, the
encoding scheme for~$\mathscr{C}$ will encode each set~$K$ through a
finite relational structure. The details are discussed in
Subsection~\ref{subsec:ellipsoid} below.

We assume that the encoding of a set~$K \subseteq \mathbb{R}^I$ carries
within it enough information to determine the set~$I$. If the encoding
also carries information about a rational~$R$ satisfying~$K \subseteq
S(0^I,R)$, then we say that~$K$ is \emph{circumscribed}, and we
write~$(K; I, R)$ to refer to it. We write~$(K; n, R)$ whenever~$I =
[n]$.

The \emph{exact feasibility problem} for~$\mathscr{C}$ takes as input
the encoding of a set~$K \subseteq \reals^I$ in~$\mathscr{C}$ and asks
for a bit~$b \in \{0,1\}$ that is~$1$ if~$K$ is non-empty, and~$0$
if~$K$ is empty. The \emph{weak feasibility problem} for~$\mathscr{C}$
takes as input the encoding of a set~$K \subseteq \reals^I$
in~$\mathscr{C}$ and a rational~$\epsilon > 0$ and asks for a bit~$b
\in \{0,1\}$ and a vector~$x \in \rationals^I$ such that:
\begin{enumerate} \itemsep=0pt
\item $b = 1$ and $x \in \Sphere{K}{\epsilon}$, or
\item $b = 0$ and $\mathrm{vol}(K) \leq \epsilon$.
\end{enumerate}
The reason why the exact feasibility problem is formulated as a
decision problem and does not ask for a feasible point is that~$K$
could well be a single point with non-rational components. In the weak
feasibility problem this is not an issue because if~$K$ is non-empty,
then the ball~$\Sphere{K}{\epsilon}$ surely contains a rational point.
The \emph{not-so-weak separation problem} for~$\mathscr{C}$ takes as
input the encoding of a set~$K \subseteq \reals^I$ in~$\mathscr{C}$, a
vector~$y \in \rationals^I$, and a rational~$\delta > 0$ and asks as
output for a bit~$b \in \{0,1\}$ and a vector~$s \in \rationals^I$
such that~$\infnorm{s} = 1$ and:
\begin{enumerate}\itemsep=0pt
\item $b = 1$ and $y \in \Sphere{K}{\delta}$, or
\item $b = 0$ and $\innprod{s}{y} + \delta \geq \sup \{ \innprod{s}{x}
  : x \in K \}$.
\end{enumerate}
The problems carry the adjective \emph{weak} in their name to stress
on the fact that in both cases the more natural requirement of
membership in~$K$ is replaced by the looser requirement of membership
in~$S(K,\gamma)$ for a given~$\gamma > 0$. For the weak separation
problem, the additional qualification \emph{not-so-weak} serves the
purpose of distinguishing it from the \emph{weak(er)} version in which
condition~2 is replaced by the looser requirement that~$b=0$
and~$\innprod{s}{y} + \delta \geq \sup \{ \innprod{s}{x} : x \in
S(K,-\delta) \}$. It turns out that the main procedure of the
ellipsoid method, as stated in the
monograph~\cite{GroetschelLovaszSchrijverBook} and in
Theorem~\ref{thm:centralcut} below, requires the \emph{not-so-weak}
version. Recall that an ellipsoid in~$\reals^I$ is a set of the
form~$E(A,a) = \{ x \in \reals^I : (x-a)^T A (x-a) \leq 1 \}$,
where~$a \in \reals^I$ is the center, and~$A$ is an~$I \times I$
positive definite matrix.

\begin{theorem}[Theorem 3.2.1 in \cite{GroetschelLovaszSchrijverBook}] \label{thm:centralcut}
There is an oracle polynomial-time algorithm, the central-cut
ellipsoid method $\CC$, that solves the following problem: Given a
rational number~$\epsilon > 0$ and a circumscribed closed convex
set~$(K; n,R)$ given by an oracle that solves the not-so-weak
separation problem for~$K$, outputs one of the following: either a
vector~$x \in \Sphere{K}{\epsilon}$, or a positive definite matrix~$A
\in \rationals^{n \times n}$ and a vector~$a \in \rationals^n$ such
that~$K \subseteq E(A,a)$ and~$\mathrm{vol}(E(A,a)) \leq \epsilon$.
\end{theorem}

We plan to use the algorithm $\CC$ from Theorem~\ref{thm:centralcut}
almost as a black box, except for the four aspects of it listed
below. Although they are not stated in Theorem~3.2.1
in~\cite{GroetschelLovaszSchrijverBook}, inspection of the proof and
the definitions in the book shows that they hold:
\begin{enumerate} \itemsep=0pt
  \item the input to the algorithm is the triple given by~$\epsilon$,~$n$
    and~$R$,
  \item the rational numbers~$\epsilon$ and~$R$ are represented in binary,
  \item the natural number~$n$ is represented in unary 
  (i.e., $2^n$ is given in binary),
  \item the algorithm makes at least one oracle query, and the output
    is determined by the answer to the last oracle call in the
    following way: if this last call was~$(y,\delta)$ and the answer
    was the pair~$(b,s)$, then~$\delta \leq \epsilon$ and the output
    vector~$x$ of $\CC$ is~$y$ itself whenever~$b = 1$, and there exists
    a positive definite matrix~$A$ and a vector~$a$ so that~$K
    \subseteq E(A,a)$ and~$\mathrm{vol}(E(A,a)) \leq \epsilon$
    whenever~$b = 0$.
  \end{enumerate}
The last point implies, in particular, that $\CC$ solves the weak
feasibility problem for the given~$K$. However, note also that the
theorem states a notably stronger claim than the existence of a
polynomial-time oracle reduction from the weak feasibility problem for
a class~$\mathscr{C}$ of sets to the not-so-weak separation problem
for the same class~$\mathscr{C}$ of sets: indeed,~$\CC$ solves the
feasibility problem for~$K$ by making oracle calls to the separation
problem for \emph{the same}~$K$.

\subsection{Definability of ellipsoid}\label{subsec:ellipsoid}

We encode sets in~$\mathscr{C}$
as finite relational structures 
in an isomorphism-invariant way.
Such encodings we call
\emph{representations}. We define this formally.

Let us first specify
what it means for two sets~$P \subseteq \reals^I$
and~$Q \subseteq \reals^J$ to be isomorphic, where~$I$ and~$J$ are two
non-empty index sets. 
For a function~$\sigma : I \rightarrow J$ and
a~$J$-vector~$v$, we denote by~$[v]^{-\sigma}$
the~$I$-vector defined by~$[v]^{-\sigma}(i) = v(\sigma(i))$ for
every~$i \in I$.
For sets of~$J$-vectors, such as~$Q$, we
define~$[Q]^{-\sigma} = \{[v]^{-\sigma} : v \in Q \}$. We say that~$P$
and~$Q$ are \emph{isomorphic}, denoted~$P \cong Q$, if there is a
bijection~$\sigma : I \rightarrow J$ such that~$P = [Q]^{-\sigma}$.
Now we can define representations of classes of sets.  A
\emph{representation} of the class~$\mathscr{C}$ of sets is a
surjective partial map~$r$ from the class of finite~$L$-structures
onto~$\mathscr{C}$, where~$L$ is a finite vocabulary with at
least one unary relation symbol~$I$, that satisfies the following
conditions:
\begin{enumerate} \itemsep=0pt
\item for every two $\mathbb{A},\mathbb{B} \in \Dom(r)$, if
  $\mathbb{A} \cong \mathbb{B}$, then $r(\mathbb{A}) \cong
  r(\mathbb{B})$,
\item for every $\mathbb{A} \in \Dom(r)$ it holds that
  $r(\mathbb{A}) \subseteq \reals^{I}$ where $I = I(\mathbb{A})$.
\end{enumerate}
A \emph{circumscribed representation} of~$\mathscr{C}$ is a surjective
partial map~$r$ from the class of finite~$L$-structures
onto~$\mathscr{C}$, where~$L$ is a finite vocabulary containing at
least one unary relation symbol~$I$ as well as a copy of the
vocabulary~$L_{\rationals}$, that satisfies the following conditions:
\begin{enumerate} \itemsep=0pt
\item for every two $\mathbb{A},\mathbb{B} \in \Dom(r)$, if
  $\mathbb{A} \cong \mathbb{B}$, then $r(\mathbb{A}) \cong
  r(\mathbb{B})$,
\item for every $\mathbb{A} \in \Dom(r)$ it holds that
  $r(\mathbb{A}) \subseteq \reals^{I}$ where $I = I(\mathbb{A})$,
\item for every~$\Astruct \in \Dom(r)$ it holds that~$r(\mathbb{A}) \subseteq \Sphere{0^I}{R}$ where~$R$ is the rational number represented by the~$L_{\rationals}$-reduct of~$\Astruct$.
\end{enumerate}
Note that a circumscribed representation of~$\mathscr{C}$ exists only
if every~$K$ in~$\mathscr{C}$ is bounded. For a given
representation~$r$ of~$\mathscr{C}$, any of the existing
preimages~$\Astruct \in r^{-1}(K)$ of a set~$K \in \mathscr{C}$ is
called a \emph{representation} of~$K$.  If~$L$ is the vocabulary of
the representation, then we say that~$\mathscr{C}$ is represented in
vocabulary~$L$. If~$\mathscr{C}$ has a representation in some
vocabulary~$L$, then we say that~$\mathscr{C}$ is a \emph{represented
  class of sets}, and if it has a circumscribed representation, then
we say that it is a \emph{represented class of circumscribed sets}.

If~$\mathscr{C}$ is a represented class of convex sets,~$L$ is the
vocabulary of the representation, and~$\Phi$ is a class of logical
formulas, then we say that the weak feasibility problem
for~$\mathscr{C}$ is~$\Phi$-definable if there exists
a~$\Phi$-interpretation that, given as input a representation of a
set~$K$ in~$\mathscr{C}$ and a rational~$\epsilon > 0$ as a structure
over~$L \disjointunion L_{\rationals}$, produces a structure
over~$L_{\bits} \disjointunion L_{\vectors}$ representing a valid
output.  It is required in addition that the represented~$K \subseteq
\reals^I$ from the input and the vector~$x \in \rationals^I$ from the
output share the same sort~$\bar{I}$ with the same relation symbol~$I$
interpreted by the same set.  Similarly, for the not-so-weak
separation problem, the input is a structure over~$L
\disjointunion L_{\rationals} \disjointunion L_{\vectors}$ and the
output is a structure over~$L_{\bits} \disjointunion
L_{\vectors}$. Again, the represented~$K \subseteq \reals^I$ and the
vector~$y \in \rationals^I$ from the input, and the vector~$s \in
\rationals^I$ from the output, share the same sort~$\bar{I}$ with the
same relation symbol~$I$ interpreted by the same set.

The following is the main result of this section.

\begin{theorem} \label{thm:definableellipsoid}
  Let~$\mathscr{C}$ be a represented class of circumscribed closed
  convex sets. If the not-so-weak separation problem for~$\mathscr{C}$
  is~$\FPC$-definable, then the weak feasibility problem
  for~$\mathscr{C}$ is also~$\FPC$-definable.
\end{theorem}

Although all the main ideas of the proof in
Subsection~\ref{subsec:proof} below were already present in
the works~\cite{Anderson:2015} and~\cite{DawarW17}, we present a
detailed proof for completeness as then the key new insights become
clearer.

At an intuitive level, the main difficulty for simulating the
ellipsoid method within a logic is that one needs to make sure that
the execution of the algorithm stays \emph{canonical}; i.e., invariant
under the isomorphisms of the input structure. The principal device to
achieve this is the following clever idea from~\cite{Anderson:2015}:
instead of running the ellipsoid method directly over the given set~$K
\subseteq \reals^I$, the algorithm is run over certain \emph{folded}
versions~$[K]^\sigma \subseteq \reals^{\sigma(I)}$ of~$K$,
where~$\sigma(I)$ is an \emph{ordered subset} of~$I$. If the execution
of the ellipsoid algorithm does not detect the difference between~$K$ and the
folded~$[K]^\sigma$, then an appropriately defined \emph{unfolding} of
the solution for~$[K]^\sigma$ will give the right solution
for~$K$. If, on the contrary, the ellipsoid detects the difference in
the form of a vector~$u \in \rationals^I$ whose folding~$[u]^\sigma$
does not unfold appropriately, then the knowledge of~$u$ is exploited
to \emph{refine} the current folding into a strictly larger
ordered~$\sigma'(I) \subseteq I$, and the execution is rebooted with
the new~$[K]^{\sigma'} \subseteq \reals^{\sigma'(I)}$. After no more
than~$|I|$ refinements the folding will be indistinguishable
from~$K$, and the execution will be correct.

The crux of the argument that makes this procedure~$\FPC$-definable is
that the ellipsoid algorithm is always operating over an
\emph{ordered} set~$\sigma(I)$. In particular, the algorithm stays
canonical, and the polynomially many steps of its execution are
expressible in fixed-point logic FP by the Immerman-Vardi
Theorem~\cite{10.1145/800070.802187, 10.1145/800070.802186}.
Indeed, the counting ability of~$\FPC$ is required only
during the folding/unfolding/refining steps.

Our formalization of these ideas will bring in two key insights that
were not present in earlier works.  The first one is the observation
that it is possible to simulate the oracle queries to the
folded~$[K]^\sigma$ by oracle queries to the original~$K$, and that
this works for any closed convex set~$K$. Furthermore,
it is possible to transfer an appropriately chosen termination
condition on the volume of the folded~$[K]^\sigma$ to the required
termination condition on the volume of the original~$K$. The
observation that the volume condition commutes with the folding
operations on arbitrary~$K$'s did not appear in the earlier work on
LPs~\cite{Anderson:2015}, nor on SDPs~\cite{DawarW17}. For LPs, the
ellipsoid method on~$[K]^\sigma$ is typically combined with a rounding
procedure to actually solve the exact feasibility problem. Therefore,
the termination condition in that case is just exact feasibility, or
plain emptiness. For SDPs,
the volume-based termination condition is not analyzed 
in~\cite{DawarW17}, since the results there
hold under the additional assumption of full-dimensionality
and the ellipsoid method always outputs a vector
in~$\Sphere{K}{\epsilon}$.
The claim that the
folding operations do have a mild effect on the volume of
arbitrary~$K$'s is the subject of Lemma~\ref{lem:volume} below.

The second key insight that our formalization brings in is the
observation that the two \emph{precision} parameters of an input for
an arbitrary~$K$, i.e., the \emph{large} radius~$R>0$ in the
circumscribing assumption, and the \emph{small} margin
guarantee~$\epsilon>0$ in the weak feasibility goal, do not interfere
with the requirement that the algorithm behaves in an
isomorphism-invariant way. Again, this observation was not clearly
analyzed in the previous work on LPs, nor on SDPs. 
As explained in the introduction, it will be
crucial for us to be able to handle arbitrarily large~$R>0$, and
arbitrarily small~$\epsilon>0$, to get the results of
Section~\ref{sec:sdp}.

Before we move on to the actual proof of
Theorem~\ref{thm:definableellipsoid}, we discuss the required material
for the method of foldings.

\subsection{Folding operations} 

Let~$I$ and~$J$ be non-empty index sets. Let~$\sigma : I \rightarrow
J$ be an onto map.  The \emph{almost-folding}~$(u)^{\sigma}$ and the
\emph{normalized almost-folding}~$(u)^{\sigma}_{\mathrm{n}}$ of
an~$I$-vector~$u$ are the~$J$-vectors defined by
\begin{equation}
(u)^{\sigma}(j) := \sum_{i \in \sigma^{-1}(j)} u(i)
\;\;\;\;\;\text{ and }\;\;\;\;\;
(u)^{\sigma}_{\mathrm{n}} := \frac{(u)^{\sigma}}{\norminf{(u)^{\sigma}}}
\end{equation}
for every~$j \in J$, with the understanding that
if~$\norminf{(u)^\sigma} = 0$, then~$(u)^\sigma_{\mathrm{n}}$ is
defined as the zero vector.  The \emph{folding}~$[u]^{\sigma}$ of
an~$I$-vector~$u$ and the \emph{unfolding}~$[v]^{-\sigma}$ of
a~$J$-vector~$v$ are the vectors defined by
\begin{equation}
[u]^{\sigma}(j) := \frac{1}{|\sigma^{-1}(j)|} \sum_{i \in \sigma^{-1}(j)} u(i)
\;\;\;\;\;\text{ and }\;\;\;\;\;
[v]^{-\sigma}(i) := v(\sigma(i))
\end{equation}
for every~$j \in J$ and every~$i \in I$, respectively.  For sets~$K
\subseteq \reals^I$ and~$L \subseteq \reals^J$, define~$[K]^{\sigma}
:= \{ [u]^{\sigma} : u \in K \}$ and~$[L]^{-\sigma} := \{
          [v]^{-\sigma} : v \in L \}$. Observe that the
          notation~$[v]^{-\sigma}$ and~$[L]^{-\sigma}$ agrees with the
          one we introduced earlier when we defined
          representations. The map~$\sigma$ is said to \emph{respect}
          a vector~$u \in \reals^I$ if~$u_i = u_{i'}$
          whenever~$\sigma(i) = \sigma(i')$ for every~$i,i' \in
          I$. The following lemma collects a few important properties
          of foldings. See Propositions~17 and~18 in~\cite{DawarW17}
          in which properties~4 and~7 from the lemma are also proved
          for all sets but stated only for convex sets. A small
          difference is that our statement of~7 is written in terms of
          the normalized almost folding operation defined above which
          is what is actually needed in the uses of the lemma.

\begin{lemma} \label{lem:props}
Let~$\sigma : I \rightarrow J$ be an onto map, let~$u$ and~$v$
be~$I$-vectors, and let~$K$ be a set of~$I$-vectors.  Then the
following hold: 
\begin{enumerate}\itemsep=0pt
\item $[au+bv]^\sigma = a[u]^\sigma + b[v]^\sigma$ for
every~$a,b \in \reals$,
\item $\normtwo{[u]^\sigma} \leq \normtwo{u}$,
\item $K \subseteq \Sphere{0^I}{R}$ implies~$[K]^\sigma \subseteq
\Sphere{0^J}{R}$,
\item $u \in \Sphere{K}{\delta}$ implies~$[u]^\sigma
\in \Sphere{[K]^\sigma}{\delta}$,
\item if~$K$ is convex, then~$[K]^\sigma$ is convex, and 
\item if~$K$ is
bounded and closed, then~$[K]^\sigma$ is bounded and closed,
\item if~$\delta > 0$
and~$\sigma$ respects~$u$, and~$\norminf{u}=1$ and~$\innprod{u}{v} +
\delta \geq \sup\{\innprod{u}{x} : x \in K\}$,
then~$\norminf{(u)^\sigma_{\mathrm{n}}}=1$
and~$\innprod{(u)^\sigma_{\mathrm{n}}}{[v]^\sigma} + \delta \geq
\sup\{\innprod{(u)^\sigma_{\mathrm{n}}}{x} : x \in [K]^\sigma\}$.
\end{enumerate}
\end{lemma}

\begin{proof}
Property 1 is straightforward by definition. 
Property 2 follows
from the inequality~$(x_1 + \cdots + x_d)^2 \leq (x_1^2 + \cdots +
x_d^2) d$, which is the special case of the Cauchy-Schwartz
inequality~$| \langle x,y \rangle | \leq \normtwo{x}\normtwo{y}$
where~$y$ is the~$d$-dimensional all-ones vector. 
Property 3 is an
immediate consequence of 2. 
Property 4 follows from 1 and 2:
if~$\normtwo{u-x} \leq \delta$, then~$\normtwo{[u]^\sigma -
  [x]^\sigma} = \normtwo{[u-x]^\sigma} \leq \normtwo{u-x} \leq
\delta$. 
Property 5 follows from
the fact that the map~$u \mapsto [u]^\sigma$ is linear.
Property 6
follows from the fact that a continuous image of a compact set is
compact: indeed the map~$u \mapsto [u]^\sigma$ is continuous, and a
subset of Euclidean space is compact if and only if it is closed and
bounded.
Property 7 follows from the straightforward fact that
whenever~$\sigma$ respects~$u$, we
have~$\innprod{(u)^\sigma}{[y]^\sigma} = \innprod{u}{y}$
and~$\norminf{(u)^\sigma} \geq
\norminf{u}$. Indeed,~$\sup\{\innprod{(u)^\sigma_{\mathrm{n}}}{x} : x
\in [K]^\sigma\} = \sup\{\innprod{(u)^\sigma_{\mathrm{n}}}{[x]^\sigma}
: x \in K \}$, and for every~$x \in K$ we
have~$\innprod{(u)^\sigma_{\mathrm{n}}}{[x]^\sigma} -
\innprod{(u)^\sigma_{\mathrm{n}}}{[v]^\sigma} =
\innprod{(u)^\sigma_{\mathrm{n}}}{[x-v]^\sigma}$.  Now
either~$\innprod{(u)^\sigma_{\mathrm{n}}}{[x-v]^\sigma}\leq 0 <
\delta$, or~$\innprod{(u)^\sigma_{\mathrm{n}}}{[x-v]^\sigma} > 0$ and
since~$\norminf{(u)^\sigma} \geq \norminf{u} = 1$, we
have~$\innprod{(u)^\sigma_{\mathrm{n}}}{[x-v]^\sigma} =
\innprod{(u)^\sigma}{[x-v]^\sigma} / \norminf{(u)^\sigma} \leq
\innprod{(u)^\sigma}{[x-v]^\sigma} = \innprod{u}{x-v} =
\innprod{u}{x}-\innprod{u}{v} \leq \delta$.  
\end{proof}

There is one further important property of foldings that we need for
correctness of the FPC-interpretation that we are about to define.
Let us extend the definition of the
set~$E(A,a) = \{ x \in \reals^J : (x-a)^T A (x-a) \leq 1 \}$ to
arbitrary positive semidefinite matrices~$A$. It should be noted that
if~$A$ is positive semidefinite but not positive definite, then at
least one of the semi-axes of~$E(A,a)$ is infinite and hence the set
is unbounded.  In this case we call~$E(A,a)$ an \emph{unbounded
  ellipsoid}.  If the simulation of the run of $\CC$ is executed until
the end over a folded~$[K]^\sigma$ and the output bit is~$0$, then the
algorithm certifies that~$[K]^\sigma$ is contained in an ellipsoid of
a small volume (see point~3 immediately following the statement of
Theorem~\ref{thm:centralcut}). To ensure that the volume of~$K$ itself
is small we use the following lemma.

\begin{lemma} \label{lem:volume} Let $K \subseteq \reals^I$ be a set,
  let~$\sigma : I \rightarrow J$ be an onto map, and let~$R \in
  \reals^{J \times I}$ and~$L \in \reals^{I \times J}$ be the matrices
  that define the linear maps~$u \mapsto [u]^\sigma$ and~$v \mapsto
  [v]^{-\sigma}$, respectively. If there is a positive definite
  matrix~$A \in \reals^{J \times J}$ and a vector~$a \in \reals^J$
  such that~$[K]^\sigma \subseteq E(A,a)$, then~$K \subseteq E(R^T A
  R, L a)$. Moreover, for every~$\epsilon > 0$ and~$r > 0$,
  if~$\mathrm{vol}(E(A,a)) \leq \epsilon$, then~$\mathrm{vol}(E(R^T A
  R, La) \cap \Sphere{0^I}{r})
\leq 2^n r^{n-1} n k \epsilon^{1/k}$, where $n = |I|$ and $k = |J|$.
\end{lemma}

\begin{proof}
Assume that~$[K]^\sigma \subseteq E(A,a)$, where~$A = B^TB$ is
positive definite. Take a point~$x \in K$.  We want to show that~$x$
is in~$E((BR)^T (BR), La)$. We have:
\begin{align}
\normtwo{BR(x-La)}^2 = \normtwo{B(Rx-RLa)}^2 = \normtwo{B(Rx-a)}^2 \leq 1,
\end{align}
with the first equality following from the linearity of~$R$, the
second equality following from the easily verified fact
that~$[[a]^{-\sigma}]^\sigma = a$, and the inequality following from
the fact that~$x \in K$ and hence~$Rx = [x]^\sigma$ belongs
to~$[K]^\sigma \subseteq E(A,a) = E(B^TB,a)$.

For the second part of the proof, observe that the matrix~$R^T A R =
(BR)^T (BR)$ is positive semidefinite. Let~$\lambda_1 \geq \cdots \geq
\lambda_n \geq 0$ be the eigenvalues of~$R^T A R$, let~$V =\{
u_1,\ldots,u_n \}$ be an orthonormal basis of corresponding
eigenvectors, and let~$(b_1, \ldots, b_n)$ be the coordinates of~$La$
with respect to the basis~$V$. The axes of symmetry of the (possibly
unbounded) ellipsoid~$E(R^T A R, La)$ correspond to the vectors
in~$V$. As we show below,~$\lambda_1 > 0$ and therefore the shortest
axis of~$E(R^T A R, La)$ has a finite length~$ 2 (1 /
\lambda_1)^{1/2}$.  It follows that~$E(R^T A R, La)$ is contained in
the set of points whose coordinates, with respect to the basis~$V$,
are given by~$[b_1 - (1 / \lambda_1)^{1/2}, b_1 + (1 /
  \lambda_1)^{1/2}] \times \reals^{n-1}$. Since
the~$r$-ball~$\Sphere{0}{r}$ is inscribed in the~$n$-dimensional
hypercube~$[-r,r]^n$, where the coordinates are again given with
respect to the basis~$V$, this implies that~$E(R^T A R, La) \cap
\Sphere{0^I}{r}$ is contained in~$[b_1 - (1 / \lambda_1)^{1/2}, b_1 +
  (1 / \lambda_1)^{1/2}] \times [-r,r]^{n-1}$.  Hence,
\begin{equation}
\mathrm{vol}(E(R^T A R, La) \cap \Sphere{0^I}{r}) \leq
2 (1 / \lambda_1)^{1/2} (2r)^{n-1} = 2^{n} r^{n-1}(1 /
\lambda_1)^{1/2}.
\end{equation}
We will finish the proof by showing that
$\mathrm{vol}(E(A,a)) \leq \epsilon$ implies $(1 / \lambda_1)^{1/2}
\leq n k \epsilon^{1/k}$,
and in particular $\lambda_1 > 0$.

Let~$\mu_1 \geq \cdots \geq \mu_{k} > 0$ be the eigenvalues of the
matrix~$A$. We have
\begin{equation}
\mathrm{vol}(E(A,a)) = V_{k} (1 / \mu_1)^{1/2} \cdots (1 /
\mu_{k})^{1/2} \geq V_{k} (1 / \mu_1)^{k/2},
\end{equation}
where~$V_k$ denotes the
volume of a~$1$-ball in the~$k$-dimensional real vector space (for the
volume of an ellipsoid see,
e.g.,~\cite{GroetschelLovaszSchrijverBook}). Therefore,
if~$\mathrm{vol}(E(A,a)) \leq \epsilon$, then~$\mu_1 \geq (V_{k} /
\epsilon)^{2/k} > k^{-2} (1/\epsilon)^{2/k}$, where the last
inequality follows from the fact that~$V_k > k^{-k}$. Now, let~$y \in
\reals^{J}$ be an eigenvector of~$A$ corresponding to the
eigenvalue~$\mu_1$, and let~$x = Ly$. Note that~$x^Tx \leq n
y^Ty$. Hence,
\begin{equation}
x^T  R^T A R x = y^T A y = \mu_1 y^T y \geq (\mu_1/n) x^T x.
\end{equation}
Since $y \not= 0$ also $x \not= 0$, and the Rayleigh quotient
principle implies that $\lambda_1 \geq \mu_1/n > 0$. Hence $\lambda_1
\geq k^{-2} (1/\epsilon)^{2/k} / n$, which gives $(1 /
\lambda_1)^{1/2} \leq n^{1/2} k \epsilon^{1/k} \leq nk\epsilon^{1/k}$.
\end{proof}

From now on, all maps~$\sigma : I \rightarrow J$ will be onto and
have~$J = [k]$ for some positive integer~$k$. Such maps define a
preorder~$\leq_\sigma$ on~$I$ with exactly~$k$ equivalence classes
which is defined by~$i \leq_\sigma i'$ if and only if~$\sigma(i) \leq
\sigma(i')$. A second map~$\sigma' : I \rightarrow [k']$ is a
\emph{refinement} of~$\sigma$ if~$\sigma'(i) \leq \sigma'(i')$
implies~$\sigma(i) \leq \sigma(i')$.  The refinement is \emph{proper}
if there exist~$i,i' \in I$ such that~$\sigma'(i) < \sigma'(i')$
and~$\sigma(i) = \sigma(i')$.  Recall that~$\sigma : I \rightarrow
[k]$ \emph{respects} a vector~$v \in \reals^I$ if~$v(i) = v(i')$
whenever~$\sigma(i) = \sigma(i')$. Since any bijective map respects
any vector, observe that if~$\sigma$ does not respect~$v$, then there
exists at least one proper refinement of~$\sigma$ that does
respect~$v$. We aim for a canonical such refinement, that we
denote~$\sigma^v$, and that is definable in~$\FPC$. We define it as
follows.

Fix an onto map $\sigma : I \rightarrow [k]$ and a vector $v \in
\reals^I$. Define:
\medskip
\begin{center}
\begin{tabular}{lclll}
$n(j)$ & $:=$ & $|\{v(\ell) : \sigma(\ell) = j \}|$ & & for $j \in [k]$, \\
$m(i)$ & $:=$ & $|\{v(\ell) : \sigma(\ell) = \sigma(i),\; 
v(\ell) \leq v(i) \}|$ & & for $i \in I$, \\
$\sigma'(i)$ & $:=$ & $n(1) + \cdots + n(\sigma(i)-1) + m(i)$ & &
for $i \in I$, \\
$k'$ & $:=$ & $n(1) + \cdots + n(k)$. & &
\end{tabular}
\end{center}
\medskip
In words,~$n(j)$ is the number of distinct~$v$-values in the~$j$-th
equivalence class of~$\leq_\sigma$, and~$m(i)$ is the number of
distinct~$v$-values in the equivalence class of~$i$ that are no bigger
than the~$v$-value~$v(i)$ of~$i$. The map~$\sigma' : I \rightarrow
[k']$ is our~$\sigma^v$. Note that if~$\sigma$ respects~$v$,
then~$\sigma^v = \sigma$.  On the other hand:

\begin{fact}
If~$\sigma$ does not respect~$v$, then~$\sigma^v$ is onto and a proper
refinement of~$\sigma$ that respects~$v$.
\end{fact}

\noindent Although not strictly needed, it is useful to note
that~$\sigma^v$ is a coarsest refinement of~$\sigma$ that
respects~$v$.  The final lemma before we proceed to the proof of
Theorem~\ref{thm:definableellipsoid} collects a few computation tasks
about foldings that are~$\FPC$-definable:

\begin{lemma} \label{lem:things}
The following operations have~$\FPC$-interpretations:
\begin{enumerate} \itemsep=0pt
\item given a set~$I$, output the~$0$ vector~$0^I$ and the
  constant~$1$ map~$\sigma : I \rightarrow [1]$,
\item given~$u \in \rationals^I$ and onto~$\sigma : I \rightarrow
  [k]$, output~$(u)^\sigma_{\mathrm{n}}$,
\item given~$u \in \rationals^k$ and onto~$\sigma : I \rightarrow
  [k]$, output~$[u]^{-\sigma}$,
\item given~$u \in \rationals^I$ and onto~$\sigma : I \rightarrow
  [k]$, output~$1$ if~$\sigma$ respects~$u$, else output $0$,
\item given~$u \in \rationals^I$ and~$\sigma : I \rightarrow [k]$,
  output~$\sigma^u : I \rightarrow [k']$.
\end{enumerate}
\end{lemma}

\begin{proof} 
All five cases are straightforward given the ability of~$\FPC$ to
perform the basic arithmetic of rational numbers, compute sums of sets
of rationals indexed by definable sets, and compute cardinalities of
definable sets.
\end{proof}

\subsection{Proof of Theorem~\ref{thm:definableellipsoid}}\label{subsec:proof}

Let~$\Psi$ be an~$\FPC$-interpretation that witnesses that the
not-so-weak separation problem for~$\mathscr{C}$ is~$\FPC$-definable.
We start by showing that there is an~$\FPC$-interpretation~$\Psi'$
that either simulates the not-so-weak separation oracle
for~$[K]^\sigma$ or outputs a vector not respected by~$\sigma$.  More
precisely,~$\Psi'$ takes as input a representation of a set~$K
\subseteq \reals^I$ in~$\mathscr{C}$, an onto mapping~$\sigma : I
\rightarrow [k]$ where~$k$ is an integer that satisfies~$1 \leq k \leq
|I|$, a vector~$y \in \rationals^k$, and a rational~$\delta > 0$ and
outputs an integer~$b \in \{-1,0,1\}$ and a vector~$s \in
\rationals^I$ such that~$\norminf{s} = 1$ and:
  \begin{enumerate} \itemsep=0pt
  \item~$b = 1$ and~$\sigma$ respects~$s$ and~$[y]^{-\sigma} \in
    S(K,\delta)$ and~$y \in S([K]^\sigma,\delta)$, or
  \item~$b = 0$ and~$\sigma$ respects~$s$
    and~$\innprod{(s)^\sigma_{\mathrm{n}}}{y} + \delta \geq \sup\{
    \innprod{(s)^\sigma_{\mathrm{n}}}{x} : x \in [K]^\sigma \}$, or
  \item~$b = -1$ and~$\sigma$ does not respect~$s$.
  \end{enumerate}
  Concretely, let~$\Psi'$ be the interpretation that does the
  following:
  \smallskip
  \begin{center}
  \begin{tabbing} 
  \;\; \= 01. \;\; \= given a representation of $K \subseteq \reals^I$
  in $\mathscr{C}$, $\sigma : I \rightarrow [k]$, $y \in \rationals^k$, 
  and $\delta \in \rationals$, \\
  \> 02. \> compute $y^- := [y]^{-\sigma}$ and $(b,s) := \Psi(K; y^-, \delta)$, \\
  \> 03. \> if $\sigma$ respects $s$, output the same $(b,s)$, \\
  \> 04. \> if $\sigma$ does not respect $s$, output $(-1,s)$.
  \end{tabbing}
  \end{center}
  \smallskip
  The claim that~$\Psi'$ is~$\FPC$-definable follows from
  points~3\ and~4\ in Lemma~\ref{lem:things}.  The claim
  that~$\Psi'$ satisfies the required conditions follows from the
  correctness of~$\Psi$, together with the fact
  that~$[[y]^{-\sigma}]^\sigma = y$, and Properties~4 and~7 in
  Lemma~\ref{lem:props}. For later use, let us note that if the
  given~$\sigma : I \rightarrow [k]$ is a bijection, then the third
  type of output~$b=-1$ cannot occur.

  Next we show how to use~$\Psi'$ in order to implement, in~$\FPC$,
  the algorithm $\CC$ from Theorem~\ref{thm:centralcut}.
  Consider the following variant $\CCC$ of $\CC$:
  \smallskip
  \begin{center}
  \begin{tabbing} 
    \;\; \= 01. \;\;\= given a rational $\epsilon > 0$ and a representation of a set
    $K \subseteq \reals^I$ in $\mathscr{C}$, \\
    \> 02. \> compute $R$ satisfying $K \subseteq \Sphere{0^I}{R}$ 
    from the representation of $K$, \\
    \> 03. \> let $n := |I|$ and $k := 1$, and let $\sigma : I \rightarrow [1]$ be the
    constant $1$ map, \\
    \> 04. \> start a run of $\CC$ on input $(\gamma,k,R)$ 
    where $\gamma := \min\{(\epsilon/(2^n R^{n-1} n k))^k,\epsilon\}$, \\
    \> 05. \> given an oracle query $(y,\delta)$, replace it by
    $(b,s) := \Psi'(K; \sigma, y, \delta)$, \\
    \> 06. \> if $\sigma$ respects $s$, then \\
    \> 07. \> \;\;\;\; \= compute $(s)^\sigma_{\mathrm{n}}$ and take the pair $(b,(s)^\sigma_{\mathrm{n}})$ as an output to the query $(y,\delta)$, \\
    \> 08. \> \> if the run of $\CC$ makes a new query $(y,\delta)$, go back to step 05, \\
    \> 09. \> \> if the run of $\CC$ makes no more queries, go to step 13, \\
    \> 10. \> else \\
    \> 11. \> \> compute $\sigma^s : I \rightarrow [k']$, the canonical 
    refinement of $\sigma$ that respects $s$, \\
    \> 12. \> \> reboot the run of $\CC$ with $\sigma := \sigma^s$ and $k := k'$ and go back to step 04, \\
    \> 13. \> let $(b,s)$ be the output of $\Psi'$ for the last oracle call $(y,\delta)$, \\
    \> 14. \> output $(b,[y]^{-\sigma})$. 
\end{tabbing}
\end{center}
\smallskip
As discussed above, $\CCC$ simulates the ellipsoid method over
folded versions~$[K]^\sigma$ of~$K$.
By Properties~3,~5 and~6 in Lemma~\ref{lem:props}
each such folded version is a circumscribed closed convex set.
A key aspect of~$\CC$ that makes this algorithm well defined is that
the only knowledge
about the targeted
set~$[K]^\sigma$
that it needs for steps~04,~05,~08 and~09 are its dimension~$k$,
its bounding radius~$R$, and correct answers to the earlier queries
to the not-so-weak separation oracle for~$[K]^\sigma$ itself.
In particular, the
algorithm does not need the folded~$[K]^\sigma$ to belong to the
class~$\mathscr{C}$.
Indeed, as long as all~$s$ vectors are respected
by a particular~$\sigma$,
Properties~4 and~7 in Lemma~\ref{lem:props}
guarantee that the
answers in step~07 stay
consistent with the
assumption that the circumscribed closed convex set given by the oracle
\emph{is}~$[K]^\sigma$. As soon as an~$s$ vector that is not respected
by~$\sigma$ is found, the map~$\sigma$ is refined, and the run
of~$\CC$ is rebooted with the new~$k$ and~$\gamma$ for the
new~$\sigma$ (and the same~$R$). 

No later than after~$|I|$
refinements of~$\sigma$, the simulation of~$\CC$ will be executed until
the end. This happens at the latest when~$\sigma$ becomes the totally
refined map: at that point~$\sigma$ is a bijection that respects
every~$s$.
Whenever a run of the simulation is completed, the algorithm reaches step~13
with a pair~$(b,s)$ and a~$\sigma$ that respects~$s$.  We use this to
show that $\CCC$ solves the weak feasibility problem for~$K$, and that
it can be implemented in~$\FPC$.

The claim that~$\CCC$ solves the weak feasibility problem for any~$K$
in~$\mathscr{C}$ is proved as follows. Let~$(b,s)$ be the output
of~$\Psi'$ for the last oracle call~$(y,\delta)$ of the execution
of~$\CC$. As noted above,~$\sigma : I \rightarrow [k]$ respects~$s$
and hence~$b \in \{0,1\}$ by Property~3 in the description of~$\Psi'$.
If~$b = 1$, then~$[y]^{-\sigma} \in \Sphere{K}{\delta}$ by Property~1
in the description of~$\Psi'$, and~$\Sphere{K}{\delta} \subseteq
\Sphere{K}{\epsilon}$ because~$\delta \leq \gamma \leq \epsilon$. This
shows that~$(b,[y]^{-\sigma})$ is a correct output for the weak
feasibility problem for~$\epsilon$ and~$K$ in case~$b = 1$. In case~$b
= 0$ we have~$[K]^\sigma \subseteq E(A,a)$ for a positive definite
matrix~$A$ and a vector~$a$, with~$\mathrm{vol}(E(A,a)) \leq \gamma
\leq (\epsilon / (2^{n} R^{n-1} n k))^k$, by point~4\ immediately
following the statement of Theorem~\ref{thm:centralcut}.  Since~$K
\subseteq \Sphere{0}{R}$, by Lemma~\ref{lem:volume} this means that
the volume of~$K$ is at most~$\epsilon$ and the answer~$b = 0$ is a
correct output.

For the implementation in~$\FPC$, we note that~$\CCC$ is a relational
WHILE algorithm that halts after at most~$|I|$ iterations all whose
steps can be computed by~$\FPC$-interpretations without
quotients. Step~01 is the description of the input. Step~02 follows
from the fact that~$K$ has a circumscribed representation: just take
the~$L_{\rationals}$-reduct of the representation of~$K$,
where~$L_{\rationals}$ is the copy of the vocabulary that is used for
representing the rational radius~$R$.  Step~03 follows from
point~1\ in Lemma~\ref{lem:things}. Step~04 follows from the
Immerman-Vardi Theorem on the fact that the representation of~$[k]$ is
an ordered structure and the computation of~$\CC$ in between oracle
calls runs in polynomial time.  Step~05 follows from the fact
that~$\Psi'$ is~$\FPC$-definable.  Step~06 is just a control
statement.  Step~07 follows from point~2\ in Lemma~\ref{lem:things}. 
Step~08 follows, again, from the Immerman-Vardi Theorem on the fact
that the representation of~$[k]$ is an ordered structure and the
computation of $\CC$ in between oracle calls runs in polynomial
time. Step~09 follows from the same reason as Step~08. Step~10 is a
control statement. Step~11 follows from point~5\ in
Lemma~\ref{lem:things}. Steps~12 and~13 are just control statements.
Step~14 follows from point~3\ in Lemma~\ref{lem:things}.

This completes the proof of Theorem~\ref{thm:definableellipsoid}, and
this section.

\section{Feasibility of SDPs}\label{sec:sdp}

In this section we use Theorem~\ref{thm:definableellipsoid} to show
that the exact feasibility of semidefinite programs is definable
in~$\Cinf{\omega}$.

\subsection{Semidefinite sets}

The \emph{semidefinite set}~$K_{A,b} \subseteq \mathbb{R}^I$ is defined
by the constraints
\begin{equation}
\langle A_i,X \rangle \leq b_i \; \text{ for } i \in M \;\text{ and }\; X \succeq 0,
\end{equation}
where~$A \in \mathbb{R}^{M \times (J \times J)}$ is an indexed set
of~$J \times J$ matrices,~$b \in \mathbb{R}^M$ is an indexed set of
reals,~$X$ is a~$J \times J$ symmetric matrix of formal
variables~$x_{ij} = x_{ji} = x_{\{i,j\}}$ for~$i,j \in J$, and~$I =
\{\{i,j \} : i,j \in J \}$ is the set of variable indices. A
\emph{circumscribed semidefinite set} is a pair~$(K_{A,b} \subseteq
\mathbb{R}^I, R)$, where~$K_{A,b} \subseteq \mathbb{R}^I$ is a
semidefinite set as defined above and~$R$ is a rational
satisfying~$K_{A,b} \subseteq S(0^I,R)$.

When~$A$ and~$b$ have rational coefficients, the semidefinite
set~$K_{A,b} \subseteq \mathbb{R}^I$ is represented by a four-sorted
structure, with one sort~$\bar{I}$ for the set~$I$ of indices of
variables, two sorts~$\bar{J}$ and~$\bar{M}$ for the index sets~$J$
and~$M$, and one sort~$\bar{B}$ for a domain~$\{0,\ldots,N-1\}$ of bit
positions that is large enough to encode all the numbers in
binary. The vocabulary~$L_{\SDP}$ includes the following relation
symbols:
\begin{enumerate} \itemsep=0pt
\item three unary symbols~$I$,~$J$ and~$M$, for~$\bar{I}$,~$\bar{J}$
  and~$\bar{M}$, respectively,
\item one ternary symbol~$P$ of type~$\bar{I} \times \bar{J} \times
  \bar{J}$,
\item one binary symbol~$\leq$ for the natural linear order
  on~$\bar{B}$,
\item three 4-ary symbols~$P_{A,s}, P_{A,n}, P_{A,d}$,
\item three binary symbols~$P_{b,s}, P_{b,n}, P_{b,d}$.
\end{enumerate}
The relation that interprets~$P$ encodes the two indices of each
variable.  The relations that interpret~$P_{A,s},P_{A,n},P_{A,d}$
encode the signs and the bits of the numerators and the denominators
of the entries of the matrices in~$\{A_i : i \in M\}$. The relations
that interpret~$P_{b,s},P_{b,n},P_{b,d}$ encode the signs and the bits
of the numerators and the denominators of the rationals in~$\{b_i : i
\in M\}$. The representation of the circumscribed semidefinite
set~$(K_{A,b} \in \mathbb{R}^I, R)$ is a structure over the
vocabulary~$L_{\SDP} \disjointunion L_{\rationals}$
whose~$L_{\SDP}$-reduct is the representation of~$K_{A,b} \in
\mathbb{R}^I$, and whose~$L_{\rationals}$-reduct is the representation
of~$R$.

The class of semidefinite sets together with the representation
defined above form a represented class of sets, which we denote
by~$\mathscr{C}_{\SDP}$. Similarly, the class of circumscribed
semidefinite sets form a represented class of circumscribed sets
denoted~$\mathscr{C}_{\SDP}^C$.

In~\cite{DawarW17} Dawar and Wang show the
FPC-definability of the 
weak optimization problem for~$\mathscr{C}_{\SDP}^C$
with the additional assumption of the input SDP being non-empty. 

\begin{theorem}[\cite{DawarW17}]\label{thm:weakoptimizationnonempty}
There exists an~$\FPC$-interpretation that takes
as input a non-empty circumscribed
semidefinite set~$(K_{A,b} \subseteq \reals^{I}, R)$, a vector $c \in \reals^{I}$
and a rational $\delta > 0$,
and outputs a vector~$x \in \reals^{I}$
such that $x \in \Sphere{K_{A,b}}{\delta}$
and $\innprod{c}{x} + \delta \geq \sup \{ \innprod{c}{y}
  : y \in K_{A,b} \}$.
\end{theorem}

\noindent In order to do so they prove
Theorem~\ref{thm:definableellipsoid} for the special case of
full-dimensional
semidefinite sets and note that weak optimization
reduces to weak feasibility by adding the cost vector as a constraint.
They also observe that a non-empty input set
can be made full-dimensional by considering
an~$\epsilon$-perturbation of the constraints, 
for an appropriately chosen~$\epsilon > 0$.
Finally, they propose an~$\FPC$-interpretation for the
not-so-weak separation oracle. 
We work out the details of a 
variant of their 
oracle
construction that takes care of a missing step in their 
proof. This fix was first published in the preliminary report of this 
paper~\cite{DBLP:journals/corr/abs-1802-02388} 
and the conference version of this paper~\cite{10.1145/3209108.3209186}.
The revised version of Wang's PhD thesis~\cite{WangPhD}
included the same fix.

As a consequence of Theorem~\ref{thm:definableellipsoid}
and the FPC-definability of the not-so-weak separation oracle 
we get the following:

\begin{theorem}\label{thm:weakfeasibility}
The weak feasibility problem for circumscribed semidefinite sets
is definable in~$\FPC$.
\end{theorem}

\subsection{Separation oracle} 

We show that the not-so-weak separation problem is~$\FPC$-definable
for the class~$\mathscr{C}_{\SDP}$ of all semidefinite sets. This
clearly implies the~$\FPC$-definability of the not-so-weak separation
problem for~$\mathscr{C}_{\SDP}^C$, which is what is needed for the
proof of Theorem~\ref{thm:weakfeasibility}. We begin with a few
definitions and lemmas.
In particular, 
since one of the steps of the separation procedure is a
reduction to a family of LPs,
we specify an encoding of an LP as a finite relational structure.

The \emph{polytope}~$K_{u,v} \subseteq \mathbb{R}^I$ is defined by a
system of linear inequalities:
\begin{equation}
\Iprod{u_i}{x} \leq v_i \;\;\text{ for } i \in M,
\end{equation}
where~$x$ is an~$I$-vector of variables,~$u \in \mathbb{R}^{M \times
  I}$ is an indexed set of~$I$-vectors, and~$v \in \mathbb{R}^M$ is an
indexed set of reals. If the entries of the vectors~$\{u_i : i \in
M\}$ and~$v$ are rational numbers, then the polytope~$K_{u,v}
\subseteq \mathbb{R}^I$ is represented by a three-sorted structure,
with two sorts~$\bar{I}$ and~$\bar{M}$ for the index sets~$I$ and~$M$,
and one sort~$\bar{B}$ for a domain~$\{0,\ldots,N-1\}$ of bit
positions that is large enough to encode all the numbers in
binary. The vocabulary~$L_{\LP}$ includes the following relation
symbols:
\begin{enumerate} \itemsep=0pt
\item two unary symbols~$I$ and~$M$, for~$\bar{I}$ and~$\bar{M}$,
  respectively,
\item one binary symbol~$\leq$ for the natural linear order
  on~$\bar{B}$,
\item three ternary symbols~$P_{u,s}, P_{u,n}, P_{u,d}$,
\item three binary symbols~$P_{v,s}, P_{v,n}, P_{v,d}$.
\end{enumerate}
The relations that interpret the symbols in point~3 encode the signs
and the bits of the numerators and the denominators of the entries of
the vectors in~$\{u_i : i \in M\}$.  The relations that interpret the
symbols in point~4 encode those of the rationals in~$\{v_i : i \in
M\}$.

Linear programs of the form:
\begin{equation}
(P) \;\;:\;\; \inf{_x} \;\Iprod{c}{x}\; \text{ s.t. } \Iprod{u_i}{x} \leq v_i \;\text{ for } i \in M,
\end{equation}
where~$x$,~$u$ and~$v$ are as specified above and~$c$ is
an~$I$-vector, are represented similarly as polytopes. The
vocabulary~$L_{\opt\LP}$ contains three additional binary
symbols~$P_{c,s},P_{c,n},P_{c,d}$ that encode the vector~$c$.

\begin{theorem}[\cite{Anderson:2015}]\label{thm:LP}
  There exists an~$\FPC$-interpretation that takes as input a linear
  program~$P : \inf{_x} \;\Iprod{c}{x}\; \text{ s.t. } \Iprod{u_i}{x}
  \leq v_i \text{ for } i \in M,$ and outputs an integer~$b \in
  \{-1,0,1\}$, a vector~$s \in \rationals^I$ and a rational~$r$, such
  that:
  \begin{enumerate} \itemsep=0pt
  \item $b = 1$ and $P$ is feasible but unbounded below, or
  \item $b = 0$ and $P$ has an optimal
    feasible solution of value~$r$, and~$s$ is one, or
  \item $b = -1$ and $P$ is infeasible.
  \end{enumerate}
\end{theorem}

\noindent We also need the following lemma from~\cite{DawarW17}
showing that the smallest eigenvalue of a given symmetric matrix can
be approximated in~$\FPC$:

\begin{lemma}[\cite{DawarW17}]~\label{lem:eigenvalue}
There exists an~$\FPC$-interpretation that takes as input a symmetric
matrix~$A \in \rationals^{I \times I}$ and a rational~$\delta > 0$ and
outputs a rational~$\hat{\lambda}$, such that~$\hat{\lambda}$ is the approximate
value of the smallest eigenvalue~$\lambda$ of~$A$ up to
precision~$\delta$, i.e.,~$|\hat{\lambda}-\lambda|\leq \delta$.
\end{lemma}

We are now ready to show the following:

\begin{proposition}
The not-so-weak separation problem for semidefinite sets
is definable in~$\FPC$.
\end{proposition}

\begin{proof}
If~$K_{A,b} \subseteq \reals^{I}$ is a non-empty semidefinite set
and~$Y \in \reals^{J \times J}$ is a symmetric matrix
outside~$K_{A,b}$, then either~$Y$ violates at least one of the linear
inequalities that describe~$K_{A,b}$, or fails to be positive
semidefinite. In the former case, we get a separating hyperplane by
taking the normal of the violated inequality, and a canonical one by
taking the sum of all of them, as in~\cite{Anderson:2015}.  In the
latter case, the smallest eigenvalue~$\lambda$ of~$Y$ is negative, and
if~$v$ is an eigenvector of this eigenvalue, then~$vv^T$ is a valid
separating hyperplane (after normalization). Such an eigenvector would
be found if we were able to find an optimal solution to the
optimization problem
\begin{equation}
\inf{_y}\; \normone{(Y-\lambda\identitymatrix)y}\ \text{
  s.t. }\ \norminf{y} = 1.
\end{equation} 
Unfortunately, this optimization problem cannot be easily phrased into
an LP because the constraint~$\norminf{y} = 1$ cannot be expressed by
linear inequalities. Here is where we differ from~\cite{DawarW17}:
first we relax the constraint~$\norminf{y} = 1$ to~$\norminf{y} \leq
1$, but then we add the condition that some component~$y_l$ is~$1$,
and we do this for each~$l \in J$ separately. Thus, for each~$l \in
J$, let~$P(Y,\lambda,l)$ be the following optimization problem:
\begin{equation}
\inf{_y}\; \normone{(Y-\lambda\identitymatrix)y}\ \text{ s.t. }\ \norminf{y} \leq 1,
\ y_l = 1.
\end{equation}
This we can formulate as an LP.  The problem~$P(Y,\lambda,l)$ may be
feasible for some~$l \in J$ and infeasible for some other~$l \in J$,
but at least one is guaranteed to be feasible. We take a solution for
each feasible one and add them together to produce a canonical
separating hyperplane. All this would be an accurate description of
what our separation oracle does if we could compute~$\lambda$ exactly, but
unfortunately only an approximation~$\hat{\lambda}$ is available.
Still, if the approximation is good enough, using~$\hat{\lambda}$ in
place of~$\lambda$ in the~$P(Y,\lambda,l)$'s will do the job.  We
provide the details.

Let~$\Psi$ be the interpretation that takes as input a symmetric
matrix~$Y \in \mathbb{Q}^{J \times J}$, a rational~$\delta > 0$, and a
representation of~$K_{A,b} \subseteq \reals^{I}$
in~$\mathscr{C}_{\SDP}$, where~$A \in \mathbb{Q}^{M \times (J \times
  J)}$ and~$b \in \mathbb{Q}^M$, and does the following:
  \smallskip
  \begin{center}
  \begin{tabbing} 
  \;\; \= 01. \;\; \= given $Y$, $\delta$, and $K_{A,b} \subseteq \reals^{I}$ as specified, \\
  \> 02. \> compute $L := \{i \in M \colon \Iprod{A_i}{Y} > b_i \}$, \\
  \> 03. \> if $|L| \neq 0$, then \\
  \> 04. \> \;\;\;\; \= compute $D := \norminf{\sum_{i \in L} A_i}$, \\
 \> 05. \> \> if $D \neq 0$, compute $S := \sum_{i \in L} A_i/D,$ and output $(0,S)$, \\
  \> 06. \> \> if $D=0$, output $(0,\identitymatrix)$, \\
   \> 07. \> else \\
   \> 08. \> \> compute $n := |J|$, \\
   \> 09. \> \> compute $\hat{\lambda}$, the smallest eigenvalue of $Y$ up to precision~${\delta/2n^2}$, \\
  \> 10. \> \> if $\hat{\lambda} > {\delta / {2n^2}}$, output $(1,\identitymatrix)$, \\
   \> 11. \> \> else \\
  \> 12. \> \> \;\;\;\; \= compute $T := \{ l \in J : P(Y,\hat{\lambda},l) \text{ is feasible with optimum } \leq \delta / 2n \}$, \\
  \> 13. \> \> \> compute $v := \{ v_l \in \rationals^{J} : l \in T \text{ and } v_l \text{ is optimal for } P(Y,\hat{\lambda},l) \}$, \\
  \> 14.  \> \> \>   compute $D := \norminf{\sum_{l \in T} v_lv_l^T}$ and $S := - \sum_{l \in T} v_lv_l^T/D$, \\
   \> 15.  \> \> \>  output $(0,S).$
       \end{tabbing}
  \end{center}
  \smallskip

Let us show that~$\Psi$ is~$\FPC$-definable and satisfies the required
conditions.

Step~01 is the description of the input. Steps~07 and~11 are control
steps. FPC-definability of Steps~02, 03, 04, 05, 06, 08, 10, 14 and~15
follows from the ability of FPC to perform the basic arithmetic of
rational numbers, compare rational numbers, and compute cardinalities
of definable sets. Step~09 follows from Lemma~\ref{lem:eigenvalue}.
Below we first argue that the output of~$\Psi$ is always correct and
finally that the Steps~12 and~13 are~$\FPC$-definable.

Suppose that~$L = \{i \in M \colon \Iprod{A_i}{Y} > b_i \} \neq
\emptyset$ and let us prove that the output in Steps~05 and~06 is
correct. If~$\sum_{i \in L} A_i$ is the zero matrix, then we have that
\begin{equation}
\sum_{i \in L} b_i < \sum_{i \in
  L} \Iprod{A_i}{Y} = \Iprod{\sum_{i \in L} A_i}{Y} = 0.
\end{equation} 
Therefore, the feasibility region~$K_{A,b}$ is empty. Indeed, every~$X
\in K_{A,b}$ satisfies
\begin{equation}
0 > \sum_{i \in L} b_i \geq \sum_{i \in L}
\Iprod{A_i}{X} = \Iprod{\sum_{i \in L} A_i}{X} = 0,
\end{equation} 
which is a contradiction. Hence, for any matrix
whose~$L_{\infty}$-norm is~$1$, in particular for the identity
matrix~$\identitymatrix$, the output~$(0,\identitymatrix)$ is correct.

If~$\sum_{i \in L} A_i$ is not the zero matrix, let~$D =
\norminf{\sum_{i \in L} A_i}$ and~$S = \sum_{i \in L} A_i/D$.  Then
for every~$X \in K_{A,b}$ we have that
\begin{equation}
\begin{aligned}
\Iprod{S}{X} & = \Iprod{{ \sum_{i \in L} { A_i \over D } }}{X} = { \frac{1}{D} } \sum_{i \in L} \Iprod{A_i}{X} 
\leq  { \frac{1}{D} } \sum_{i \in L} b_i <  { 1 \over D } \sum_{i \in L} \Iprod{A_i}{Y} = \Iprod{S}{Y}. 
\end{aligned}
\end{equation}
Moreover, the matrix~$S$ has~$L_{\infty}$-norm~$1$. So the output is
correct.

Suppose that~$L = \{i \in M \colon \Iprod{A_i}{Y} > b_i \} =
\emptyset$,~$n = |J|$ and~$\hat{\lambda} > {\delta / {2n^2}}$, and let
us argue that the output in Step~10 is correct. Observe that, for
every~$i \in M$, the matrix~$Y$ satisfies~$\Iprod{A_i}{Y} \leq b_i$,
and its smallest eigenvalue~$\lambda$ is positive, which means that
the matrix~$Y$ is positive semidefinite. Hence,~$Y$ is in the
feasibility region~$K_{A.b}$ and the output is correct.

Finally, let us assume that~$\hat{\lambda} \leq {\delta / {2n^2}}$. In
this case, for every~$l \in J$, the~$\FPC$ interpretation needs to
compute the optimal value and an optimal solution of the optimization
problem~$P(Y,\hat{\lambda},l)$. To show that this is possible, we
define an essentially equivalent linear program~$P'(l)$ and use
Theorem~\ref{thm:LP} to conclude.

To perform Steps~12 and~13 the~$\FPC$ interpretation takes, for
each~$l \in J$, the linear program~$P'(l)$ with variables~$\{x_i : i
\in J\} \cup \{y_i : i \in J\}$, defined by:
\medskip
\begin{center}
\begin{tabular}{llll}
 $\inf{_{x,y}}$ & \; & $\textstyle{\sum_{i \in J} x_i}$ & \\
 s.t. & & $-x_i \leq  (Yy - \hat{\lambda}y)_i   \leq x_i$, & for every $i \in J$ \\
 & & $-1 \leq y_i \leq 1$, & for every  $i \in J$ \\
 & & $y_l = 1$, & 
\end{tabular}
\end{center}
\medskip
where~$y$ is the vector~$\{y_i : i \in J\}$.  In the following,
since~$Y$ and~$\hat{\lambda}$ are fixed, let us write~$P(l)$ instead
of~$P(Y,\hat{\lambda},l)$.

\begin{claim}
The program~$P(l)$ is feasible if and only if the program~$P'(l)$ is
feasible and the optimal values of~$P(l)$ and~$P'(l)$ are the same.
Moreover, if a vector~$\{x_i : i \in J\} \cup\{y_i : i \in J\}$ is an
optimal solution to~$P'(l)$, then the vector~$\{y_i : i \in J\}$ is an
optimal solution to~$P(l)$.
\end{claim}

\begin{proof}
Suppose that the feasibility region of~$P(l)$ is non-empty. For every
vector~$y = \{y_i : i \in J\}$ in the feasibility region of~$P(l)$,
the vector~$\{x_i : i \in J\} \cup\{y_i : i \in J\}$, where~$x_i =
|(Yy - \hat{\lambda}y)_i|$, belongs to the feasibility region
of~$P'(l)$ and its value~$\sum_{i \in J} x_i = \normone{(Y -
  \hat{\lambda} I)y}$ is the same as the value of~$\{y_i : i \in J\}$
for~$P(l)$. Therefore, the feasibility region of~$P'(l)$ is non-empty
and the optimal value~$opt'$ of~$P'(l)$ is smaller or equal to the
optimal value~$opt$ of~$P(l)$.

Suppose that the feasibility region of~$P'(l)$ is non-empty, and take
an optimal solution~$\{x_i : i \in J\} \cup\{y_i : i \in J\}$
for~$P'(l)$.
It holds that~$\norminf{y} = 1$ and~$y_l = 1$, so the vector~$y$ is in
the feasibility region of~$P(l)$. Therefore, the feasibility region
of~$P(l)$ is non-empty, and~$opt \leq \normone{(Y - \hat{\lambda}
  I)y}~$. Moreover, for every~$i \in J$, we have that~$|(Yy -
\hat{\lambda}y)_i | \leq x_i$ so~$ \normone{(Y - \hat{\lambda} I)y}
\leq \sum_{i \in J} x_i = opt'$. On the other hand we know that~$opt'
\leq opt$. To summarize
\begin{equation}
opt \leq \normone{(Y - \hat{\lambda} I)y} \leq \sum_{i \in J} x_i  = opt' \leq opt.
\end{equation} 
Hence, the vector~$y$ is an optimal solution for~$P(l)$ and the
optimal values are the same.
\end{proof}

To perform Steps~12 and~13 the~$\FPC$ interpretation computes, for
every~$l \in J$, an optimal solution and the optimal value of the
optimization problem~$P(l)$, by computing an optimal solution and the
optimal value of the linear program~$P'(l)$ via Theorem~\ref{thm:LP},
and projecting the output to the variables~$ \{y_i : i \in J\}$.

We now show that the set~$T$, defined in Step~12, is non-empty, and
that~$\norminf{\sum_{l \in T}v_l v_l^T} \neq 0$. It follows that the
output matrix~$S$ in Step~14 is well defined.

\begin{claim}
$T \neq \emptyset$.
\end{claim}

\begin{proof}
Let~$v$ be an eigenvector of~$Y$ with the smallest
eigenvalue~$\lambda$, and let~$\norminf{v} = 1$. We have the following
\begin{equation}
\begin{aligned}
 \normone{(Y - \hat{\lambda}I)v } & = \normone{(Y - \lambda I)v - (\hat{\lambda} - \lambda)Iv  } \leq \\ & \leq \normone{ (Y - \lambda I)v } + \normone{ (\hat{\lambda} - \lambda)Iv  } = \\  
& =  \normone{ (\hat{\lambda} - \lambda)Iv  }   \leq {\delta \over {2n^2}} n \norminf{ v } = {\delta \over 2n}.
\end{aligned}
\end{equation}
If there exists~$l \in J$ such that~$v_l = 1$, then~$v \in P(l)$
and~$T \neq \emptyset$. Otherwise, there exists~$l \in J$ such
that~$v_l = -1$. Then~$-v \in P(l)$ and we are done as well.
\end{proof}

\begin{claim}\label{claim:inequality}
$1 \leq \norminf{\sum_{l \in T}v_l v_l^T} \leq |T|$.
\end{claim}

\begin{proof}
Observe that for every~$l \in J$ all the main diagonal entries of the
matrix~$v_l v_l^T$ are squares and since~$\norminf{v_l} = 1$, at least
one of those entries is equal~$1$. Therefore,
\begin{equation}
\norminf{\sum_{l \in T}v_l v_l^T} \geq 1,
\end{equation}
and on the other hand,
\begin{equation}
\norminf{\sum_{l \in T}v_l v_l^T} \leq \sum_{l \in T} \norminf{v_l v_l^T} = |T|.
\end{equation}
\end{proof}

Finally, let us show that the output $(0,S)$ in Step~15 is correct.

\begin{claim}\label{claim:optimal}
For every $l \in T$, let~$v_l$ be the optimal solution of~$P(l)$. Then for every $X \in K_{A,b}$, 
\begin{equation}
\Iprod{-v_lv_l^T}{Y} + {\delta \over n} \geq \Iprod{-v_lv_l^T}{X}.
\label{eqn:firstineqq}
\end{equation} 
\end{claim} 

\begin{proof} 
Take $X \in K_{A,b}$. Since the matrix~$X$ is positive semidefinite, $\Iprod{-v_lv_l^T}{X} = - v_l^T X v_l \leq 0$. We will show that $\Iprod{-v_lv_l^T}{Y} + {\delta / n} \geq 0$. It holds that
\begin{equation}
\begin{aligned}
 \Iprod{-v_lv_l^T}{Y}  & = -v_l^T Y v_l = -v_l^T (\hat{\lambda}I + (Y - \hat{\lambda}I )) v_l  = \\ 
& = - \hat{\lambda} v_l^T v_l - v_l^T (Y - \hat{\lambda}I ) v_l  \geq
- \hat{\lambda} v_l^T v_l - |v_l^T (Y - \hat{\lambda}I ) v_l| \geq \\ 
& \geq  - \hat{\lambda} v_l^T v_l - \norminf{v_l} \normone{(Y - \hat{\lambda}I ) v_l} \geq - \hat{\lambda} v_l^T v_l - {\delta \over 2n} .
\end{aligned}
\end{equation}
It follows that  
\begin{equation}
\Iprod{-v_lv_l^T}{Y} + {\delta \over n} \geq - \hat{\lambda} v_l^T v_l + {\delta \over 2n}.
\end{equation} 

Now if~$\hat{\lambda} \leq 0$, then~$- \hat{\lambda} v_l^T v_l +
{\delta / 2n} = - \hat{\lambda} \normtwo{v_l}^2 + {\delta / 2n} \geq
{\delta / 2n} > 0$.  Otherwise~$0 < \hat{\lambda} \leq {\delta /
  2n^2}$, and
\begin{equation}
\hat{\lambda} v_l^T v_l  \leq {\delta \over 2n^2} \normtwo{v_l}^2 \leq {\delta \over 2n^2} (\sqrt{n}\norminf{v_l})^2 = {\delta \over 2n^2} n  = {\delta \over 2n}.
\end{equation}
Hence, $ - \hat{\lambda} v_l^T v_l + {\delta / 2n} \geq - {\delta / 2n} + {\delta / 2n} = 0$.
\end{proof}

We finish the proof by showing that for every $X \in K_{A,b}$, 
\begin{equation}
\Iprod{S}{Y} + {\delta} \geq \Iprod{S}{X}.
\end{equation} 
Let~$X$ be any matrix in~$K_{A,b}$. From now on, let~$D =
\norminf{\sum_{l \in T} v_lv_l^T}$. Recall from
Claim~\ref{claim:inequality} that~$1 \leq D \leq |T|$.  It holds that
\begin{equation}
\begin{aligned}
\Iprod{S}{Y} & = \Iprod{{- \sum_{l \in T} \frac{v_lv_l^T}{D} }}{Y} = {1 \over D} \sum_{l \in T} \Iprod{- v_lv_l^T}{Y} \geq {1 \over D} \sum_{l \in T} (\Iprod{- v_lv_l^T}{X} - {\delta \over n}) = \\
& =  \Iprod{{- \sum_{l \in T} \frac{v_lv_l^T}{D} }}{X} - {|T| \over D} {\delta \over n}
 = \Iprod{S}{X} - {|T| \over n}{\delta \over D} \geq  \Iprod{S}{X} - \delta,
\end{aligned}
\end{equation}
where the first inequality follows from~\eqref{eqn:firstineqq} in
Claim~\ref{claim:optimal} and the last inequality follows from the
fact that~$|T| \leq n$ and~$D \geq 1$.
\end{proof}

\subsection{Exact feasibility}

We use Theorem~\ref{thm:weakfeasibility} to prove the main result of
this section:

\begin{theorem}\label{thm:feasibility}
The exact feasibility problem for semidefinite sets
is definable in~$\Cinf{\omega}$.
\end{theorem}

We begin the proof by relating the problem of exact feasibility to the
subject of Theorem~\ref{thm:weakfeasibility}, i.e., the weak
feasibility problem for circumscribed semidefinite sets.

For any~$R > 0$ and any semidefinite set~$K_{A,b}$,
the~$R$-\emph{restriction} of~$K_{A,b}$
is the set of all those points in~$K_{A,b}$
whose~$L_{\infty}$-norm is bounded by~$R$, i.e., it is the
semidefinite set given by:
\smallskip
\begin{center}
\begin{tabular}{lllll}
$\langle A_i,X \rangle$ & $\leq$ & $b_i$ &  & for  $i \in M$, \\
$X_{\{i,j\}}$ & $\leq$ &  $R$ &  & for  $i,j \in J$, \\
$-X_{\{i,j\}}$ & $\leq$ & $R$ &  & for $i,j \in J$, \\ 
$X \succeq 0$. & & 
\end{tabular}
\end{center}
\smallskip

For any~$\epsilon > 0$ and any semidefinite set~$K_{A,b}$,
the $\epsilon$-\emph{relaxation} of~$K_{A,b}$
is the semidefinite set given by:
\begin{center}
\begin{tabular}{lllll}
$\langle A_i,X \rangle$ & $\leq$ & $b_i+\epsilon$ & & for $i \in M$ \\ 
$X \succeq 0$.
\end{tabular}
\end{center}

Since an~$R$-restriction of a semidefinite set is
a semidefinite set itself, it makes sense to talk about 
its $\epsilon$-relaxation.
The question of emptiness for~$\epsilon$-relaxations
of~$R$-restrictions of semidefinite sets is closely linked to the
exact feasibility problem under consideration.
Recall the Cantor Intersection Theorem: If~$K_1 \supseteq K_2
\supseteq \cdots$ is a decreasing nested sequence of non-empty compact
subsets of~$\reals^n$, then the intersection~$\bigcap_{i \geq 1} K_i$
is non-empty. We use it for the following lemma.

\begin{lemma}\label{lem:limit}
A semidefinite set~$K_{A,b}$ is non-empty if and only if there exists
a positive rational~$R$ such that for every positive
rational~$\epsilon$ it holds that the~$\epsilon$-relaxation of
the~$R$-restriction of~$K_{A,b}$ is non-empty.
\end{lemma}

\begin{proof}
Assume that~$K_{A,b}$ is non-empty and let~$x$ be a point in it.
Let~$R$ be a rational bigger than~$\norminf{x}$. Then~$x$ is also in
the~$R$-restriction of~$K_{A,b}$, and therefore in
the~$\epsilon$-relaxation of the~$R$-restriction of~$K_{A,b}$ for
every positive rational~$\epsilon$.

Assume now that~$R$ is a positive rational such that
the~$\epsilon$-relaxation of the~$R$-restriction of~$K_{A,b}$ is
non-empty for every positive rational~$\epsilon$. For each positive
integer~$m$, let~$K_m$ be the~$1/m$-relaxation of the~$R$-restriction
of~$K_{A,b}$.  Each~$K_m$ is closed and bounded, hence
compact. Moreover~$K_1 \supseteq K_2 \supseteq \cdots$, i.e., the
sets~$K_m$ form a decreasing nested sequence of non-empty subsets
of~$\reals^I$. It therefore follows from the Cantor Intersection
Theorem that~$\bigcap_{m \geq 1} K_m$ is non-empty. The claim follows
from the observation that~$\bigcap_{m \geq 1} K_m$ is indeed
the~$R$-restriction of~$K_{A,b}$.
\end{proof}

It follows from Theorem~\ref{thm:weakfeasibility} that the emptiness
problem for~$\epsilon$-relaxations of~$R$-restrictions of semidefinite
sets is definable in~$\FPC$ in the following sense.

\begin{proposition}\label{prop:epsilonR}
There exists a formula~$\psi$ of~$\FPC$ such that if~$\Astruct$ is a
structure over~$L_{\SDP} \disjointunion L_{\mathbb{Q}} \disjointunion
L_{\mathbb{Q}}$, representing a semidefinite set~$K_{A,b} \subseteq
\mathbb{R}^I$ and two positive rational numbers~$R$ and~$\epsilon$,
then:
\begin{enumerate} \itemsep=0pt
\item if~$\Astruct \models \psi$, then the~$\epsilon$-relaxation of
  the~$R$-restriction of~$K_{A,b}$ is non-empty,
\item if~$\Astruct \not \models \psi$, then the~$R$-restriction
  of~$K_{A,b}$ is empty.
\end{enumerate}
\end{proposition}

\begin{proof}
Let~$\Phi$ be an~$\FPC$-interpretation that witnesses that the weak
feasibility problem for the class of circumscribed semidefinite sets
is FPC-definable.  The formula~$\psi$ takes as input the
representation of a semidefinite set~$K_{A,b} \subseteq \mathbb{R}^I$,
a rational~$\epsilon > 0$ and a rational~$R > 0$, and does the
following:
\smallskip
  \begin{center}
  \begin{tabbing} 
  \;\; \= 01. \;\; \= given $K_{A,b} \subseteq \mathbb{R}^I$, $\epsilon$ and $R$ as specified, \\
  \> 02. \> compute $k : = |I|$, \\
  \> 03. \> compute $R' := \lceil \sqrt{k(R+{\epsilon})^2} \rceil$, \\
      \> 04. \> compute a representation of $K$, the ${\epsilon}$-relaxation of the $R$-restriction of $K_{A,b}$, \\
      \> 05. \> compute $m := \max\; \{ \normtwo{A_i} : i \in M \} \cup \{1\}$, \\
       \> 06. \> compute $\delta = \epsilon^k/(k!(2km)^{k})$, \\
  \> 07. \> compute $(b,x) := \Phi((K,R'),\delta)$, \\
    \> 08. \> if $b = 1$ output $\top$, \\
  \> 09. \> if $b = 0$ output $\bot$.
  \end{tabbing}
  \end{center}
  \smallskip
  This procedure is clearly FPC-definable. In order to prove
  correctness we will need the following lemma.

\begin{lemma}\label{lem:ball}
Let~$A \in \mathbb{R}^{M \times (J \times J)}$,~$b \in
\mathbb{R}^{M}$,~$I = \{ {\{i,j \} : i,j \in J} \}$,~$k = |I|$, and~$m
= \max\;\{\normtwo{A_i} : i \in M \} \cup \{1\}$.  For any~$\epsilon >
0$, if the semidefinite set~$K_{A,b} \in \mathbb{R}^I$ is non-empty,
then its~$\epsilon$-relaxation has volume greater than
$\delta = \epsilon^k/(k! (2km)^k)$.
\end{lemma}

\begin{proof}
Take~$\epsilon_1 = \epsilon/2km$. Let~$Y$ be an element
of~$K_{A,b}$. We will show that the~$\epsilon_1$-ball around~$Y +
\epsilon_1 \identitymatrix$ is included in the~$\epsilon$-relaxation
of~$K_{A,b}$. It will follow that the volume of
the~$\epsilon$-relaxation of~$K_{A,b}$ is at least~$\epsilon_1^kV_k$,
where~$V_k$ is the volume of a~$1$-ball around~$Y$ in
the~$k$-dimensional real vector space. Since~$V_k > {1/k!}$ this
finishes the proof.

Suppose that~$T \in S(Y + \epsilon_1 \identitymatrix,
\epsilon_1)$. This means that~$T = Y + \epsilon_1 \identitymatrix +
Z$, where~$\normtwo{Z} \leq \epsilon_1$. We start by showing that~$T$
is positive semidefinite. Let~$v$ be a vector whose~$L_{2}$-norm
is~$1$.  It holds that
\begin{equation}
\begin{aligned}
v^TTv & = v^T(Y + \epsilon_1 \identitymatrix + Z)v = v^TYv + \epsilon_1v^T\identitymatrix v + v^TZv \geq \\
& \geq 0 + \epsilon_1 \normtwo{v}^2 + \Iprod{vv^T}{Z} \geq \epsilon_1 - | \Iprod{vv^T}{Z} | \geq \\
& \geq \epsilon_1 - \normtwo{vv^T} \normtwo{Z} = \epsilon_1 - \normtwo{v}^2 \normtwo{Z} \geq \epsilon_1 - \epsilon_1 = 0.
\end{aligned}
\end{equation}
Moreover, for every~$i \in M$, we have
\begin{equation}
\begin{aligned}
 \langle A_i,  T\rangle - b_i & =  \Iprod{A_i}{Y} + \Iprod{A_i}{\epsilon_1 \identitymatrix} + \Iprod{A_i}{Z} - b_i  \leq \\
& \leq  \Iprod{A_i}{\epsilon_1 \identitymatrix} + \Iprod{A_i}{Z}  \leq |\Iprod{A_i}{\epsilon_1 \identitymatrix} | + |\Iprod{A_i}{Z} | \leq \\
& \leq  \epsilon_1 \normtwo{A_i} \normtwo{\identitymatrix} +  \normtwo{A_i}  \normtwo{Z}  \leq \\
& \leq  \epsilon_1 \normtwo{A_i} \sqrt{k}+   \normtwo{A_i} \epsilon_1  = \\
& =  {\epsilon \normtwo{A_i}  \sqrt{k} \over 2km }  +   { \normtwo{A_i}  \epsilon \over 2km } \leq \\
 & \leq {\epsilon \over 2 \sqrt{k}} + {\epsilon \over 2k} \leq \epsilon,
\end{aligned}
\end{equation}
where the first inequality follows from the fact that~$Y$ is an
element of~$K_{A,b}$ hence, for every~$i \in M$, it satisfies~$\langle
A_i,Y \rangle \leq b_i$, the third inequality follows from the
Cauchy-Schwartz inequality~$| \Iprod{x}{y} | \leq \normtwo{x}
\normtwo{y}$, the second to last inequality follows from the fact
that~$m = \max\; \{ \normtwo{A_i} : i \in M \} \cup \{1\}$, and the
last inequality follows from the fact that~$k = |I| \geq 1$.
\end{proof}

  We are now ready to conclude the proof. Observe that
  the~$L_{\infty}$-norm of any point that belongs to
  the~$\epsilon$-relaxation of the~$R$-restriction of a semidefinite
  set is bounded by~$R+\epsilon$, therefore the pair~$(K,R')$ computed
  in Steps~03 and~04 is a representation of a circumscribed
  semidefinite set. Let~$(b,x)$ be the pair computed in Step~07.
 
If~$b=1$, then there exists a point in~$S(K,\delta)$, which in
particular means that~$K$ is non-empty, so the output in Step~08 is
correct. If~$b=0$, then we know that the volume of~$K$ is at
most~$\delta$. The inequalities that define~$K$ have the
form~$\Iprod{A_i}{X} \leq b_i+\epsilon$ for~$i \in M$,
and~$X_{\{i,j\}} \leq R+\epsilon$ or~$-X_{\{i,j\}} \leq R+\epsilon$
for~$i,j \in J$. The maximum 2-norm of the normals of these
inequalities and~$1$ is~$m = \max\;\{\normtwo{A_i} : i \in M \} \cup
\{1\}$, so Lemma~\ref{lem:ball} applies. This means that~$K$ is empty,
and the output in Step~09 is correct.
\end{proof}

To finish the proof of Theorem~\ref{thm:feasibility} we show a
technical lemma that may sound a bit surprising at first:
it sounds as if it was stating that~$\Cinf{k}$-definability is closed
under second-order quantification over unbounded domains, which cannot
be true. However, on closer look, the lemma states this \emph{only} if
the vocabularies of the quantified and the body parts of the formula
are totally disjoint. In particular, this means that the domains of
the sorts in the quantified and body parts of the formula stay
unrelated \emph{except} through the counting mechanism of~$\Cinf{k}$.

Note for the record that if~$L$ and~$K$ are two many-sorted
vocabularies with disjoint sorts, then obviously the vocabulary~$L
\cup K$ does not contain any relation symbol whose type mixes the
sorts of~$L$ and~$K$.  If~$\mathscr{A}$ is a class of~$L\cup
K$-structures and~$\mathscr{B}$ is a class of~$K$-structures, we use
the notation~$\exists \mathscr{B}\cdot\mathscr{A}$ to denote the class
of all finite~$L$-structures~$\Astruct$ for which there exists a
structure~$\Bstruct \in \mathscr{B}$ such that~$\Astruct
\disjointunion \Bstruct \in \mathscr{A}$. In words, this is the set of
finite structures that can be disjointly extended and expanded by a
structure in~$\mathscr{B}$ to a structure in~$\mathscr{A}$. Similarly,
we use~$\forall \mathscr{B}\cdot\mathscr{A}$ to denote the class of
all finite~$L$-structures~$\Astruct$ such that for all
structures~$\Bstruct \in \mathscr{B}$ we have that~$\Astruct
\disjointunion \Bstruct \in \mathscr{A}$.  In words, this is the set
of finite structures all whose disjoint extensions and expansions by a
structure in~$\mathscr{B}$ are in~$\mathscr{A}$.

\begin{lemma}\label{lem:plugging}
Let~$L$ and~$K$ be many-sorted vocabularies with disjoint sorts,
let~$\mathscr{A}$ be a class of finite~$L \cup K$-structures, and
let~$\mathscr{B}$ be a class of
finite~$K$-structures. If~$\mathscr{A}$ is~$\Cinf{k}$-definable, then
the classes of~$L$-structures~$\exists \mathscr{B} \cdot \mathscr{A}$
and~$\forall \mathscr{B} \cdot \mathscr{A}$ are
also~$\Cinf{k}$-definable.
\end{lemma}

\begin{proof}
The proof is a simple \emph{Booleanization} trick to replace the
finite quantifiers~$\exists^{\geq i}$ over the sorts in~$K$ by finite
propositional formulas, followed by replacing~$\exists \mathscr{B}$
and~$\forall \mathscr{B}$ by infinite disjunctions and conjunctions,
respectively, indexed by the structures in~$\mathscr{B}$. We provide
the details.  Let~$\phi$ be a formula of the many-sorted vocabulary~$L
\cup K$ with all variables of the~$L$-sorts among~$x_1,\ldots,x_k$,
and all variables of the~$K$-sorts among~$y_1,\ldots,y_k$.  Note that
since~$L$ and~$K$ have disjoint sorts, all the atomic subformulas
of~$\phi$ have all its variables among~$x_1,\ldots,x_k$ or all its
variables among~$y_1,\ldots,y_k$. In other words, there are no atomic
subformulas with mixed~$x$-$y$ variables.  For every
finite~$K$-structure~$\Bstruct$ with domain~$B$ and every~$b =
(b_1,\ldots,b_k) \in B^k$, let~$\phi(\Bstruct,b)$ be the
\emph{Booleanization} of~$\phi$ with respect to the atomic
interpretation of~$K$ given by~$\Bstruct$, the domain of
quantification~$B$ for the variables of the~$K$-sorts, and the
free-variable substitution~$x := b$.  Formally, using the
notation~$[E]$ for the truth value of the statement~$E$, the
formula~$\phi' = \phi(\Bstruct,b)$ is defined inductively:
\medskip
\begin{enumerate} \itemsep=0pt
\item if $\phi = R(x_{i_1},\ldots,x_{i_\ell})$ with $R \in L \cup \{=\}$, define $\phi' = \phi$,
\item if $\phi = R(y_{i_1},\ldots,y_{i_\ell})$ with $R \in K \cup \{=\}$, define $\phi' = [(b_{i_1},\ldots,b_{i_\ell}) \in R(\Bstruct)]$,
\item if $\phi = \neg \theta$, define $\phi' = \neg \theta(\Bstruct,b)$,
\item if $\phi = \bigwedge_i \theta_i$, define $\phi' = \bigwedge_i \theta_i(\Bstruct,b)$,
\item if $\phi = \exists^{\geq t} x_i (\theta)$, define $\phi' = \exists^{\geq t} x_i (\theta(\Bstruct,b))$,
\item if $\phi = \exists^{\geq t} y_i (\theta)$, define
\begin{equation}
\phi' = \bigvee_{c \in B^t} \Big({\bigwedge_{j,j' \in [t] \atop j \not= j'} [c_j \not= c_{j'}] \wedge \bigwedge_{j \in [t]} \theta(\Bstruct,b[i/c_j])}\Big).
\end{equation}
\end{enumerate}
Since there are no atomic subformulas with mixed~$x$-$y$ variables,
the definition covers all cases.  The construction
of~$\phi(\Bstruct,b)$ was designed so that for every finite~$(L \cup
K)$-structure~$\Cstruct$ with~$L$- and~$K$-reducts~$\Astruct$
and~$\Bstruct$ with domains~$A$ and~$B$, respectively, every~$a \in
A^k$ and every~$b \in B^k$, it holds that~$\Cstruct \models \phi[a,b]$
if and only if~$\Astruct \models \phi(\Bstruct,b)[a]$.  Now, if~$\phi$
is an~$(L \cup K)$-sentence, define~$\phi(\Bstruct) := \bigvee_{b \in
  B^k} \phi(\Bstruct,b)$ and~$\phi^\exists := \bigvee_{\Bstruct \in
  \mathscr{B}} \phi(\Bstruct)$.  It follows from the definitions
that~$\phi^\exists$ defines~$\exists \mathscr{B} \cdot
\mathscr{A}$. Similarly, defining~$\phi^\forall := \bigwedge_{\Bstruct
  \in \mathscr{B}} \phi(\Bstruct)$ works for~$\forall \mathscr{B}
\cdot \mathscr{A}$.
\end{proof}

We put everything together in the proof of Theorem~\ref{thm:feasibility}.

\begin{proof}[Proof of Theorem~\ref{thm:feasibility}]
Let~$\psi$ be the~$L_{\SDP} \disjointunion L_{\mathbb{Q}}
\disjointunion L_{\mathbb{Q}}$-formula of~$\FPC$ defined in
Proposition~\ref{prop:epsilonR}.
Let~$l$ be the number of variables in~$\psi$.  By the translation
from~$l$-variable~$\FPC$ to~$\Cinf{2l}$ (see
Section~\ref{sec:preliminaries}), there exists an~$L_{\SDP}
\disjointunion L_{\mathbb{Q}} \disjointunion
L_{\mathbb{Q}}$-formula~$\tau$ of~$\Cinf{2l}$ defining the same
class~$\mathscr{A}$ of finite structures.
The vocabulary of~$\mathscr{A}$ is a disjoint union of~$L_{\SDP}$ and
two copies of~$L_{\mathbb{Q}}$.  Hence, all three of those
vocabularies have disjoint sorts.  Let~$\mathscr{B}_R$ be the class of
finite structures which are representations of positive rational
numbers over the first copy of~$L_{\mathbb{Q}}$, and
let~$\mathscr{B}_\epsilon$ be the class of finite structures which are
representations of positive rational numbers over the second copy
of~$L_{\mathbb{Q}}$.
By Lemma~\ref{lem:plugging} the class~$\forall \mathscr{B}_\epsilon
\cdot \mathscr{A}$, and hence~$\exists \mathscr{B}_R \cdot \forall
\mathscr{B}_\epsilon \cdot \mathscr{A}$, is
also~$\Cinf{2l}$-definable. Let~$\phi$ be the~$L_{\SDP}$-formula
of~$\Cinf{2l}$ defining this last class. Lemma~\ref{lem:limit} implies
that~$\phi$ defines the exact feasibility problem for semidefinite
sets.
\end{proof}

\section{Sums-of-Squares Proofs and the Lasserre Hierarchy} \label{sec:sos}

In this section we develop the descriptive complexity of the problem
of deciding the existence of degree-bounded SOS proofs.  Along the way
we discuss the relationship between the Lasserre hierarchy of SDP
relaxations and SOS refutations, and how~$0/1$-valued variables ensure
that they satisfy strong duality. These results will be used in the
next section. The strong duality property 
will also imply that SOS refutations exist if and only if
\emph{approximate} SOS refutations exist, for a notion of approximate
SOS refutation that we introduce.  This will be used to complement the
descriptive complexity results for SOS proofs by getting a stronger
upper bound in the case of refutations.

\subsection{Descriptive Complexity of SOS Proofs}

We begin with a few definitions.
Let~$x_1,\ldots,x_n$ be a set of variables. In the following whenever
we talk about polynomials or monomials we mean polynomials and
monomials over the set of variables~$x_1, \ldots, x_n$ and real or
rational coefficients.
By $B_n$ we denote the set containing the following polynomials:
\begin{equation}\label{eq:axioms}
1, \ \ \ x_i, \ \ \ 1-x_i, \ \ \ x_i^2 - x_i, \ \ \ x_i - x_i^2, \ \ \ \text{ for every } i \in [n].
\end{equation}
We refer to the inequalities~$p \geq 0$ for~$p \in B_n$ as
\emph{Boolean axioms}. A polynomial~$s$ is a \emph{sum of squares of
polynomials} if it has the form~$s = \sum_{i \in [l]} r_i^2$, for some
polynomials~$r_1, \ldots, r_l$.  For a set
of polynomials~$Q$ and a
polynomial~$p$, a \emph{Sums-of-Squares (SOS) proof} of~$p \geq 0$
from~$Q$ is 
an indexed set of polynomials~$\{ s_q : q \in \bar{Q} \}$ that satisfy
an identity
\begin{equation}\label{eqn:sos}
\sum_{q \in \bar{Q}} q s_q = p,
\end{equation}
where, $\bar{Q}= Q \cup B_n$
and for every $q \in \bar{Q}$, the polynomial $s_q$ is a sum of squares
of polynomials.
 The degree of the proof is defined as~$\max \{
\deg(qs_q) : q \in \bar{Q}\}$, where, for a polynomial~$p$, the
notation~$\deg(p)$ denotes the degree of~$p$.

One should think about the set of
polynomials~$Q$ as representing a system of polynomial
inequalities~$\{ q \geq 0 : q \in Q \}$.  The
identity~(\ref{eqn:sos}) implies that any~$0/1$-solution to this
system satisfies also the inequality~$p \geq 0$.  Therefore, if~$p =
-1$, a proof certifies that the system~$\{ q \geq 0 : q \in Q\}$
has no~$0/1$-solutions. This is why we call it a \emph{refutation}
of~$Q$.  
The definition of SOS as a proof system is sometimes attributed to
Grigoriev and Vorobyov~\cite{GrigorievVorobyov2001}. We note that, in
the case of refutations, our definition is the special case their
Positivstellensatz \cite[Definition 2]{GrigorievVorobyov2001} in which
all non-trivial products of~$q_j$'s have~$0$ multipliers. Unlike
theirs, our proof system includes the Boolean axioms by default, thus
ensuring completeness even for proofs of polynomial inequalities
over~$0/1$-valued variables.

We consider the problem of deciding the existence of~$\SOS$ proofs and
refutations of a fixed degree~$2d$ for a set of polynomials given as
input. The first easy observation is that the proof-existence problem
can be reduced to the exact feasibility problem for semidefinite sets,
and the reduction can be done in~$\FPC$.  Then we ask whether the
exactness condition in the feasibility problem for semidefinite sets
can be relaxed, and we achieve this for refutations.  In other words:
\begin{enumerate} \itemsep=0pt
\item Proof-existence reduces in~$\FPC$ to exact feasibility for
  semidefinite sets.
\item Refutation-existence reduces in~$\FPC$ to weak feasibility for
  semidefinite sets.
\end{enumerate}
We note that, in both cases, the semidefinite sets in the outcome of
this reduction are not circumscribed.  Roughly, this is because their
elements are vectors encoding sequences of coefficients of sum of
squares polynomials that form a valid proof. By the results
in~\cite{DBLP:conf/innovations/ODonnell17}, the bit-complexity of
these coefficients may not be polynomially bounded, not even for
proofs over Boolean variables~\cite{DBLP:conf/icalp/RaghavendraW17},
nor for refutations over Boolean
variables~\cite{DBLP:conf/lics/Hakoniemi21}.

As stated, point 1.\ above is
almost a reformulation of the problem. In order to prove point 2.\ we
need to develop a notion of approximate refutation, and combine it
with a strong duality theorem that characterizes the existence of SOS
refutations in terms of
so-called
\emph{pseudoexpectations}~\cite{DBLP:conf/stoc/BarakBHKSZ12}. We note
that the strong duality theorem that we need relies on the assumption
that the Boolean axioms are 
included in the definition of SOS proof.

Finally, we combine these~$\FPC$ reductions with the results of the
previous section in order to get the following:

\begin{corollary}\label{cor:decidingproofs}
For every fixed positive integer~$d$, the problems of deciding the
existence of~SOS proofs of degree~$2d$, and~SOS refutations of
degree~$2d$, are~$\Cinf{\omega}$-definable.  Moreover, there exists a
constant~$c$, independent of~$d$, such that the defining formulas are
in~$\Cinf{cd}$.
\end{corollary}

As usual with descriptive complexity results like these, we need to
fix some encoding of the input as finite relational structures.  In
this case the inputs are indexed sets of polynomials, where each
polynomial is an indexed set of monomials and coefficients.  The exact
choice of encoding is not very essential, but we propose one for
concreteness.

Let~$I$ be an index set for variables and let~$\{x_i : i \in I\}$ be a
set of formal variables. A \emph{monomial} is a product of
variables. For~$\alpha = (\alpha_i : i \in I) \in \naturals^I$, we use
the notation~$x^\alpha$ to denote the monomial that has
\emph{degree}~$\alpha_i$ on variable~$x_i$. We write~$|\alpha|$ for
the degree~$\sum_{i \in I} \alpha_i$ of the monomial~$x^\alpha$.  A
polynomial is a finite linear combination of monomials, i.e., a formal
expression of the form~$\sum_\alpha c_\alpha x^\alpha$ in which all
but finitely many of the coefficients~$c_\alpha$ are zero. A
polynomial~$p$ with rational coefficients is represented by a
three-sorted structure, with a sort~$\bar{I}$ for the index set~$I$, a
second sort~$\bar{M}$ for the finite set of monomials that have
non-zero coefficients in~$p$, and a third sort~$\bar{B}$ for a
domain~$\{0,\ldots,N-1\}$ of bit positions, where~$N$ is large enough
to encode all the coefficients of~$p$ and all the degrees of its
monomials in binary. The vocabulary
of this structure has one unary relation symbol~$I$ for~$\bar{I}$, one
binary relation symbol~$\leq$ for the natural linear order
on~$\bar{B}$, three binary relations symbols~$P_s$,~$P_n$, and~$P_d$
of type~$\bar{M} \times \bar{B}$ that encode, for each monomial, the
sign, the bits of the numerator, and the bits of the denominator of
its coefficient, respectively, and a ternary relation symbol~$D$ of
type~$\bar{M} \times \bar{I} \times \bar{B}$ that encodes, for each
monomial and each variable, the bits of the degree of this variable in
the monomial.

Let~$J$ be an index set for polynomials and let~$\{ p_j : j \in J \}$
be a set of polynomials on the variables~$\{x_i : i \in I\}$. Such a
set is represented by a four-sorted structure, with a sort~$\bar{J}$
for the index set~$J$, and the three sorts~$\bar{I},\bar{M},\bar{B}$
of the previous paragraph. The vocabulary for this structure has one
unary relation symbol~$J$ for~$\bar{J}$, one binary relation
symbol~$\leq$ for the natural linear order on~$\bar{B}$, three ternary
relation symbols~$P_s$,~$P_n$, and~$P_d$ of
type~$\bar{J} \times \bar{M} \times \bar{B}$ that encode, for
each~$j \in J$, the coefficients of the monomials in~$p_j$, and a
four-ary relation symbol of
type~$\bar{J} \times \bar{M} \times \bar{I} \times \bar{B}$ that
encodes, for each~$j \in J$, the degrees of the variables in the
monomials in~$p_j$.

\subsection{The Lasserre hierarchy}\label{subsec:lasserre}

There is a sense in which sums-of-squares proofs can be seen as the
\emph{dual solutions} in a hierarchy of semidefinite programming
relaxations of an associated optimization problem. This correspondence
will be used explicitly in~Subsection~\ref{subsec:approx}.
Some of the
concepts we introduce now will also be useful in the next
Subsection~\ref{subsec:sos-as-sdps}.

We adopt the setting in~\cite{Josz2016}. For a set of
polynomials~$\{q_0, q_1, \ldots, q_k\}$, 
 we denote the following polynomial
optimization problem by~$\POP(q_0;\{q_1, \ldots, q_k\})$:
\begin{equation}
(\POP) \;\;:\;\; \inf{_x} \;q_0(x)\; \text{ s.t. } q_i(x) \geq 0 \;\text{ for }  i \in [k].
\end{equation}

Take a positive integer~$d$. Recall that we use the
notation~$x^\alpha$, where~$\alpha = (\alpha_i : i \in [n]) \in \naturals^n$,
to denote the monomial that has degree~$\alpha_i$ on
variable~$x_i$. We identify the monomial~$x^\alpha$ with its vector of
degrees~$\alpha$. By~$M_d$ we denote the matrix indexed by monomials
of degree at most~$d$ 
defined by~$(M_d)_{\alpha,\beta} = x^{\alpha+\beta}$. For every
monomial~$x^\alpha$, we introduce a variable~$y_\alpha$ and
by~$M_d(y)$ we denote the corresponding matrix of variables, defined
by~$(M_d(y))_{\alpha,\beta} = y_{\alpha+\beta}$. More generally, for
any polynomial~$q = \sum_\gamma c_\gamma x^\gamma$, the
matrix~$M_{q,d}$, indexed by monomials of degree at most~$d$, is
defined by~$M_{q,d} = qM_d$,
i.e.,~$(M_{q,d})_{\alpha,\beta} = qx^{\alpha+\beta}$. The
corresponding matrix of variables~$M_{q,d}(y)$ is defined
by~$(M_{q,d}(y))_{\alpha,\beta} = \sum_{\gamma} c_\gamma
y_{\alpha+\beta+\gamma}$. Observe that the entries of the
matrix~$M_{q,d}$ are polynomials of degree at
most~$2d+ \deg(q)$, while
the entries of the matrix~$M_{q,d}(y)$ are the corresponding linear
combinations of variables. Note also that~$M_{1,d} = M_d$
and~$M_{1,d}(y) = M_d(y)$.  For every variable~$y_{\alpha}$, consider
the coefficients of~$y_{\alpha}$ in the matrix~$M_{q,d}(y)$. Those
coefficients form a matrix which we denote
by~$A_{q,d,\alpha}$. Formally, for~$|\alpha| \leq 2d + \deg(q)$, the
matrices~$A_{q,d,\alpha}$ are defined as the real matrices
satisfying~$M_{q,d}(y)= \sum_\alpha y_\alpha A_{q,d,\alpha}$ or
equivalently~$M_{q,d} = \sum_\alpha x^\alpha A_{q,d,\alpha}$. Finally,
for any polynomial~$q$, by~$d_q$ we denote the biggest integer
satisfying~$2d_q + \deg(q) \leq 2d$.

Let~$Q$ be a set of polynomials and
let~$q_0 = \sum_\alpha x^\alpha$ be a polynomial. For any
positive integer~$d$, the \emph{level}-$d$
\emph{Lasserre SDP relaxation} of the
polynomial optimization
problem~$\POP(q_0;Q)$ is
the pair of semidefinite programs~$(P_d, D_d)$, where~$P_d$ is the
\emph{primal} semidefinite program:
\begin{equation}
\begin{array}{llll}
\inf{_y} & & \textstyle{\sum_\alpha a_\alpha y_\alpha} \\
\text{s.t.} & & y_\emptyset = 1 \\
& & M_{q,d_q}(y) \succeq 0 & \text{ for } q \in Q
\end{array}
\end{equation}
and~$D_d$ is the \emph{dual} semidefinite program:
\begin{equation}
\begin{array}{llll}
 \sup{_{z,Z}} & & z \\
 \text{s.t.} & & \textstyle{\sum_{q \in Q} \Iprod{A_{q,d_q,\emptyset}}{Z_q} = a_\emptyset -z} \\
 & & \textstyle{\sum_{q \in Q} \Iprod{A_{q,d_q,\alpha}}{Z_q} = a_\alpha} & \text{ for } 
1 \leq |\alpha| \leq 2d \\
 & & Z_q \succeq 0 \text{ for } q \in Q
\end{array}
\end{equation}

Weak SDP duality implies that the optimal value of~$P_d$ is always
greater or equal than the optimal value of~$D_d$. The main theorem
in~\cite{Josz2016} establishes a condition which guarantees strong
duality for primal and dual SDP problems in the Lasserre hierarchy.

\begin{theorem}[\cite{Josz2016}]\label{thm:strongduality}
If~$\POP(q_0;Q)$ is a polynomial optimization problem where one of the
inequalities describing the feasibility region is~$R^2 - \sum_{i \in
  [n]} x_i^2 \geq 0$, then for every positive integer~$d$, the optimal
values of~$P_d$ and~$D_d$ are equal.
\end{theorem}

\noindent Strong duality for primal and dual problems implies, in
particular, that~$P_d$ is infeasible if and only if~$D_d$ is unbounded
above and, analogously,~$P_d$ is unbounded below if and only if~$D_d$
is infeasible.

The polynomial optimization problem~$\POP(q_0;Q)$ is called
\emph{encircled} if a polynomial~$R^2 - \sum_{i \in [n]} x_i^2$ can be
obtained as a non-negative linear combination of polynomials from~$Q$
of degree at most~$2$. The following lemma implies strong duality for
primal and dual SDP problems in the Lasserre hierarchy for encircled
polynomial optimization problems.

\begin{lemma}\label{lem:linearcombination}
Let~$Q$ be a set of polynomials and let~$p = \sum_{q \in Q} c_q q$ be
a non-negative linear combination of polynomials from~$Q$, such
that~$\deg(p) = \max \{ \deg(q) : c_q > 0 \}$.  For some
polynomial~$q_0$, let~$(P_d, D_d)$ and~$(P'_d, D'_d)$ be the level-$d$
Lasserre SDP relaxations of~$\POP(q_0;Q)$ and~$\POP(q_0;Q \cup
\{p\})$, respectively. The optimal values of~$P_d$ and~$P'_d$, as well
as the optimal values of~$D_d$ and~$D'_d$ are equal.
\end{lemma}

\begin{proof}
Let~$q_0 = \sum_\alpha a_\alpha x^\alpha$ and let~$d$ be some positive
integer.

The primal~$P'_d$ is the following semidefinite program:
\begin{equation}
\begin{array}{llll}
 \inf{_y} & & \textstyle{\sum_\alpha a_\alpha y_{\alpha}} \\
 \text{s.t.} & & y_{\emptyset} = 1 \\
 & & M_{q,d_q}(y) \succeq 0 & \text{ for } q \in Q \\
 & & M_{p,d_p}(y) \succeq 0
\end{array}
\end{equation}

Let~$P = \{ q \in Q : c_q > 0 \}$. Note that since~$\deg(p) = \max \{
\deg(q) : q \in P \}$, for every~$q \in P$, we have~$d_p \leq d_q$. For
each~$q \in P$, by~$M'_{q,d_q}(y)$ let us denote the principal
submatrix of~$M_{q,d_q}(y)$ obtained by removing the rows and columns
indexed by monomials of degree greater than~$d_p$.  Observe
that~$M_{p,d_p}(y) = \sum_{q \in P} c_q M'_{q,d_q}(y)$.  Since the
constraints~$\{ M_{q,d_q}(y) \succeq 0 : q \in P \}$ imply the
constraint~$M_{p,d_p}(y) = \sum_{q \in P} c_q M'_{q,d_q}(y) \succeq
0$, the feasibility regions, and therefore also the optimal values,
of~$P_d$ and~$P'_d$ are the same.

The dual $D'_d$ is the following semidefinite program:
\begin{equation}
\begin{array}{llll}
 \sup{_{z,Z}} & & z \\
 \text{s.t.} & & \textstyle{\sum_{q \in Q} \Iprod{A_{q,d_q,\emptyset}}{Z_q} + \Iprod{A_{p,d_p,\emptyset}}{Z_p}= a_\emptyset-z} \\
 & & \textstyle{\sum_{q \in Q} \Iprod{A_{q,d_q,\alpha}}{Z_q} + \Iprod{A_{p,d_p,\alpha}}{Z_p}= a_\alpha} & \text{ for } 1 \leq |\alpha| \leq 2d \\
 & & Z_q \succeq 0 & \text{ for } q \in Q \\
 & & Z_p \succeq 0
\end{array}
\end{equation}

Any solution to the program~$D_d$ can be extended to a solution to the
program~$D'_d$ with the same optimal value by taking~$Z_p$ to be the
zero matrix. On the other hand, any
solution~$(z, \{Z_q\}_{q \in Q}, Z_p)$ to the program~$D'_d$ gives
rise to a solution~$(\tilde z, \{\tilde Z_q\}_{q \in Q})$ to the
program~$D_d$ with the same optimal value by setting~$\tilde z := z$
and~$\tilde Z_q := Z_q + c_q Z_p$ for each~$q \in P$,
and~$\tilde Z_q := Z_q$ for each~$q \in Q \setminus P$. This follows
from the fact
that~$A_{p,d_p,\alpha} = \sum_{q \in P} c_q A_{q,d_q,\alpha}$.
\end{proof}

\subsection{SOS proofs as semidefinite sets} \label{subsec:sos-as-sdps}

Fix a set of polynomials~$Q$ and a further
polynomial~$p = \sum_\alpha a_\alpha x^\alpha$ such
that~$\deg(p) \leq 2d$.  Our goal now is to describe
degree-$2d$~$\SOS$ proofs of the polynomial inequality~$p \geq 0$
from~$Q$ as points in a semidefinite set~$K_d(Q,p)$ that we are about
to define.
  Recall that a
degree-$2d$~$\SOS$ proof of~$p \geq 0$ from~$Q$ is
an indexed set of polynomials~$\{ s_q : q \in \bar{Q} \}$ that satisfy
 an identity
$\sum_{q \in \bar{Q}} q s_q = p$ where
$\bar{Q} = Q \cup B_n$ and for every $q \in \bar{Q}$, the
polynomial $s_q$ is a sum of squares of polynomials and has degree at
most~$2d_q$.  A polynomial~$s$ of degree at most~$2t$ is a sum of
squares if and only if there exists a positive semidefinite matrix~$Z$
indexed by monomials of degree at most~$t$ such
that~$s = \Iprod{M_t}{Z}$. Therefore, there exists a
degree-$2d$~$\SOS$ proof of the polynomial inequality~$p \geq 0$
from~$Q$ if and only if, for every~$q \in \bar{Q}$, there exists a
positive semidefinite matrix~$Z_q$ indexed by monomials of degree at
most~$d_q$ such that
\begin{equation}\label{eq:proof}
\sum_{q \in \bar{Q}} q \Iprod{M_{d_q}}{Z_q} = \sum_\alpha a_\alpha x^\alpha.
\end{equation}

Let us have a closer look at the expression~$\sum_{q \in \bar{Q}} q \Iprod{M_{d_q}}{Z_q}$ on the left-hand side of the
above identity. It can be rewritten in terms of the 
matrices introduced at the beginning of Subsection~\ref{subsec:lasserre}:
\begin{equation}\label{eq:matrices}
\begin{aligned}
\sum_{q \in \bar{Q}} q \Iprod{M_{d_q}}{Z_q} & = \sum_{q \in \bar{Q}} \Iprod{M_{q,d_q}}{Z_q}   =  \sum_{q \in \bar{Q}} \Iprod{\sum_\alpha x^\alpha A_{q,d_q,\alpha}}{Z_q} = \\ & =  \sum_{\alpha} x^\alpha \sum_{q \in \bar{Q}} \Iprod{A_{q,d_q,\alpha}}{Z_q}.
\end{aligned}
\end{equation}
Hence, there exists a
degree-$2d$~$\SOS$ proof of~$p \geq 0$
from~$Q$ if, and only if, 
there exists a set of positive semidefinite
matrices~$\{Z_q : q \in \bar{Q} \}$ such that 
for every~$q \in \bar{Q}$ the matrix~$Z_q$ is
indexed by monomials of degree at
most~$d_q$ and for all~$|\alpha| \leq 2d$ it holds
$\sum_{q \in \bar{Q}} \Iprod{A_{q,d_q,\alpha}}{Z_q}= a_\alpha$,
which, in turn, can be expressed as non-emptiness of the
semidefinite set~$K_d(Q,p) \subseteq \mathbb{R}^{I_d}$ given by:
\begin{equation}
 \sum_{q \in \bar{Q}} \Iprod{A_{q,d_q,\alpha}}{Z_q}= a_\alpha \; \text{ for }  |\alpha| \leq 2d \;\text{ and }\; X \succeq 0,
\end{equation}
where~$J_d = \{(q,x^\alpha) : q \in \bar{Q}, |\alpha| \leq d_q \}$ is a set of indices,~$X$ is a~$J_d \times J_d$ symmetric matrix of formal variables,~$I_d = \{ \{ ( q,x^\alpha ), ( q',x^{\alpha'} ) \} : (q,x^\alpha), (q',x^{\alpha'}) \in J_d  \}$ 
is a set of variable indices, and for every~$q \in \bar{Q}$, the matrix~$Z_q$ is the principal submatrix of~$X$ corresponding to the rows and columns indexed by~$\{(q,x^\alpha) : |\alpha| \leq d_q \}$.

Indeed, from every feasible point~$X \in K_d(Q,p)$ we get a set of
positive semidefinite matrices~$\{Z_q : q \in \bar{Q} \}$ satisfying
the identity~(\ref{eq:proof}) by setting~${Z}_q$ be the principal
submatrix of~${X}$ corresponding to the rows and columns indexed
by~$\{(q,x^\alpha) : |\alpha| \leq d_q \}$. On the other hand, any set
of positive semidefinite matrices~$\{Z_q : q \in \bar{Q} \}$
satisfying the identity~(\ref{eq:proof}) can be extended to a point
in~$K_d(Q,p)$ by setting all remaining variables to~$0$.

The representation of the semidefinite set~$K_d(Q,p)$ can be easily
obtained from the representation of the set of polynomials~$Q$ and the
polynomial~$p$ by means of~$\FPC$-interpretations:

\begin{fact}\label{fact:reduction}
For every fixed positive integer~$d$, there is
an~$\FPC$-interpretation that takes a set of polynomials~$Q$ and a
polynomial~$p$ as input and outputs a representation of the
semidefinite set~$K_d(Q,p)$.  Moreover, there exists a constant~$c$,
independent of~$d$, such that the formulas in the~$\FPC$
interpretation have at most~$cd$ variables.
\end{fact}

Therefore, as a consequence of Theorem~\ref{thm:feasibility} we obtain
Corollary~\ref{cor:decidingproofs}.

\begin{proof}[Proof of Corollary~\ref{cor:decidingproofs}]
Let us fix a positive integer~$d$ and let~$\Phi$ be the
FPC-interpretation from Fact~\ref{fact:reduction}. We compose~$\Phi$
with the~$\Cinf{\omega}$-sentence from Theorem~\ref{thm:feasibility}
that decides the exact feasibility of semidefinite sets.  The
resulting sentence~$\psi$ decides the existence of an~$\SOS$ proof of
degree~$2d$.  It is a sentence of~$\Cinf{k}$, where~$k = cd$, for an
integer~$c$ that is independent of~$d$. A~$\Cinf{\omega}$-sentence
deciding the existence of an~$\SOS$ refutation of degree~$2d$ is
obtained analogously by starting with an~$\FPC$-interpretation which
takes as input a set of polynomials~$Q$ and outputs the semidefinite
set~$K_d(Q,-1)$.
\end{proof}

\subsection{SOS refutations} \label{subsec:approx}

We will now relate the existence of~$\SOS$ refutations of a set of
polynomials~$Q$ to the primal and dual problems in the Lasserre
hierarchy for the polynomial optimization problem~$\POP(0;\bar{Q})$.
Then we will introduce the concept of~$\epsilon$-approximate SOS refutation
and use the primal-dual correspondence to show that,
for small enough~$\epsilon > 0$, the existence of SOS refutations is
equivalent to the existence of~$\epsilon$-approximate ones. It will
follow from this that the problem of deciding the existence of~$\SOS$
refutations of a fixed degree reduces, by means
of~$\FPC$-interpretations, to the weak feasibility problem for
semidefinite sets.

For any set of polynomials~$Q$, the polynomial optimization
problem~$\POP(0;\bar{Q})$,
characterizing the existence of~$0/1$-solutions
to the system of polynomial inequalities~$\{ q \geq 0 : q \in Q\}$, 
will be denoted by~$\Sol(Q)$:
\begin{equation}
(\Sol(Q)) \;\;:\;\; \inf{_x} \;0\; \text{ s.t. } q(x) \geq 0 \;\text{ for }  q \in \bar{Q}.
\end{equation}
Indeed, the optimization problem~$\Sol(Q)$ is feasible if and only if
the system of polynomial inequalities~$\{ q \geq 0 : q \in Q\}$ has
a~$0/1$-solution if and only if
the optimal value of~$\Sol(Q)$ is~$0$.
Otherwise, the optimal value of~$\Sol(Q)$
is~$+\infty$. 
Although we care only about the feasibility of~$\Sol(Q)$,
we define it as an optimization problem, since
we want to analyze its Lasserre
SDP relaxations.

For a positive integer~$d$, by~$(P_d(Q), D_d(Q))$ we denote the
level-$d$ Lasserre
SDP relaxation of the polynomial optimization
problem~$\Sol(Q)$, i.e.,~$P_d(Q)$ is the semidefinite program:
\begin{equation}
\begin{array}{llll}
 \inf{_y} & & 0 \\
\text{s.t.} & & y_{\emptyset} = 1 \\
& & M_{q,d_q}(y) \succeq 0 & \text{ for } q \in \bar{Q}
\end{array}
\end{equation}
and~$D_d(Q)$ is the semidefinite program:
\begin{equation}
\begin{array}{llll}
 \sup{_{z,Z}} & & z \\
\text{s.t.} & & \textstyle{\sum_{q \in \bar{Q}} \Iprod{A_{q,d_q,\emptyset}}{Z_q} = -z} \\
& & \textstyle{\sum_{q \in \bar{Q}} \Iprod{A_{q,d_q,\alpha}}{Z_q} = 0} & \text{ for } 1 \leq |\alpha| \leq 2d \\
& & Z_q \succeq 0 & \text{ for } q \in \bar{Q} 
\end{array}
\end{equation}

Observe that degree-$2d$~$\SOS$ refutations of~$Q$ correspond
precisely to the feasible solutions to~$D_d(Q)$ with value~$1$
(see~(\ref{eq:proof}) and~(\ref{eq:matrices})).
The following lemma summarizes the
relationship between degree-$2d$~$\SOS$ refutations of~$Q$ and
solutions to the program~$D_d(Q)$. The second equivalence follows from
the fact that by multiplying a solution to~$D_d(Q)$ with value~$v$ by
any~$c \geq 0$ we obtain another solution with value~$cv$.

\begin{lemma}\label{lem:refutation}
There exists an~$\SOS$ refutation of~$Q$ of degree~$2d$ if and only
if~$D_d(Q)$ has a solution with value~$1$ if and only if the optimal
value of~$D_d(Q)$ is~$+\infty$.
\end{lemma}

For a system of polynomials~$Q$, a \emph{pseudoexpectation for~$Q$ of
  degree~$2d$} is a linear mapping~$F$ from the set of polynomials of
degree at most~$2d$ over the set of variables~$x_1, \ldots, x_n$ to
the reals such that~$F(1) = 1$, and for every~$q \in \bar{Q}$ and
every sum of squares polynomial~$s$ of degree at most~$2d_q$, we
have~$F(qs) \geq 0$.

A linear mapping from the set of polynomials of degree at most~$2d$ to
the reals is uniquely defined by its restriction to
monomials. Therefore, there is a natural one-to-one correspondence
between linear functions from the set of polynomials of degree at
most~$2d$ to the reals and assignments to the set of variables~$\{
y_{\alpha} : |\alpha| \leq 2d \}$ of the program~$P_d(Q)$, given
by~$G(y_\alpha) = F(x^\alpha)$. We recall 
the known fact that an assignment~$G$ to the
variables of~$P_d(Q)$ is a feasible solution if and only if~$F$ is a
pseudoexpectation of degree~$2d$.

\begin{lemma}\label{lem:pseudoexpectation}
There exists a degree-$2d$ pseudoexpectation for~$Q$ if and only if
the program~$P_d(Q)$ is feasible.
\end{lemma}

\begin{proof}
Let~$F$ be a linear function from the set of polynomials of degree at
most~$2d$ to the reals and let~$G$ be the corresponding assignment to
the variables of~$P_d(Q)$.  The statement of the lemma follows by
showing that for every~$q \in \bar{Q}$, the matrix~$M_{q,d_q}(G(y))$
is positive semidefinite if and only if for every sum of squares
polynomial~$s$ of degree at most~$2d_q$, we have~$F(qs) \geq 0$.

Let us take some~$q \in \bar{Q}$. Observe that for every matrix~$Z$
indexed by monomials of degree at most~$d_q$, we have
\begin{equation}
\Iprod{M_{q,d_q}(G(y))}{Z} =\Iprod{F(M_{q,d_q})}{Z})= F(
\Iprod{qM_{d_q}}{Z}) = F(q \Iprod{M_{d_q}}{Z}).
\end{equation}
The
matrix~$M_{q,d_q}(G(y))$ is positive semidefinite if and only if for
every positive semidefinite matrix~$Z$ indexed by monomials of degree
at most~$d_q$, it holds that $\Iprod{M_{q,d_q}(G(y))}{Z} = F(q
\Iprod{M_{d_q}}{Z}) \geq 0$ if and only if~$F(qs) \geq 0$ for every
sum of squares polynomial~$s$ of degree at most~$2d_q$. The last
equivalence follows from the fact that a polynomial~$s$ of degree at
most~$2t$ is a sum of squares if and only if there exists a positive
semidefinite matrix~$Z$ indexed by monomials of degree at most~$t$
such that~$s = \Iprod{M_t}{Z}$.
\end{proof}

Note that by summing the
inequalities~$1-x_1 \geq 0, \ldots, 1-x_n \geq 0$, together with the
inequalities~$x_1 - x_1^2 \geq 0, \ldots, x_n - x_n^2 \geq 0$, we get
the inequality~$n - \sum_{i \in [n]} x^2 \geq 0$, which witnesses the
fact that the problem~$\Sol(Q)$ is encircled. By
Lemma~\ref{lem:linearcombination} and Theorem~\ref{thm:strongduality}
it follows that for the problem~$\Sol(Q)$ there is no duality gap
between primal and dual SDP problems in the Lasserre hierarchy.  In
particular, the optimal value of~$D_d(Q)$ is~$+\infty$ if and only
if~$P_d(Q)$ is infeasible.  Now, recall from
Lemma~\ref{lem:refutation} that the optimal value of~$D_d(Q)$
is~$+\infty$ if and only if there exists an SOS refutation of~$Q$ of
degree~$2d$, and from Lemma~\ref{lem:pseudoexpectation} that the
program~$P_d(Q)$ is infeasible if and only if there is no
pseudoexpectation for~$Q$ of degree~$2d$. Hence, we obtain the
following:

\begin{corollary}\label{cor:strongduality}
There exists an $\SOS$ refutation of~$Q$ of degree~$2d$ if and only if
there is no pseudoexpectation for~$Q$ of degree~$2d$.
\end{corollary}

For any~$\epsilon > 0$, an~$\epsilon$-\emph{approximate}
degree-$2d$~$\SOS$ refutation of a set of polynomials~$Q$ is
an indexed set of polynomials~$\{ s_q : q \in \bar{Q} \}$ that satisfy
an identity
\begin{equation}
\sum_{q \in \bar{Q}}qs_q = \sum_{\alpha}a_\alpha x^\alpha,
\end{equation}
where for every~$q \in \bar{Q}$, the polynomial~$s_q$ is a sum of
squares, for each~$x^\alpha$ of degree at least~$1$ we
have~$|a_\alpha| \leq \epsilon$,
and~$|1 + a_{\emptyset}| \leq \epsilon$.  In the same way as the
degree-$2d$~$\SOS$ refutations correspond to the points in the
semidefinite set~$K_d(Q,-1)$, the~$\epsilon$-approximate
degree-$2d$~$\SOS$ refutations correspond to the points in
the~$\epsilon$-relaxation of~$K_d(Q,-1)$.

In what follows, suppose that~$Q$ has no degree-$2d$~$\SOS$
refutation. By Corollary~\ref{cor:strongduality}, there exists a
degree-$2d$ pseudoexpectation. This in turn, as we will show now,
precludes even the existence of~$\epsilon$-approximate refutations,
for small enough~$\epsilon$. The key is the following lemma, which
says that in the presence of Boolean axioms the absolute values of a
pseudoexpectation on the set of monomials are bounded by~$1$.

\begin{lemma}
If~$F$ is a degree-$2d$ pseudoexpectation for~$Q$, then~$0 \leq F(m)
\leq 1$ for every monomial~$m$ of degree at most~$d$, and~$-1 \leq
F(m) \leq 1$ for every monomial~$m$ of degree at most~$2d$.
\end{lemma}

\begin{proof}
  Consider a monomial~$m$ written as a product of powers of distinct
  variables.  The \emph{multilinearization}~$\bar{m}$ of~$m$ is the
  monomial obtained from~$m$ by replacing each~$x^c$ with~$c \geq 2$
  that appears in this product by~$x$. For instance, the
  multilinearization of~$x^2y^3z$ is the monomial~$xyz$.

First we show that if~$m$ is a monomial of degree at most~$2d$,
then~$F(\bar{m}) =
F(m)$. We do this by showing that~$F(x^2m) = F(xm)$ for every
variable~$x$ and every monomial~$m$ of degree at most~$2d-2$. Fix such
a monomial~$m$ and let~$r$ and~$s$ be monomials of degree at
most~$d-1$ such that~$m = rs$. Note that~$m = p^2 - q^2$ where~$p =
(r+s)/2$ and~$q = (r-s)/2$, and both~$p^2$ and~$q^2$ have degree at
most~$2d-2$. It holds that

\begin{equation}\label{eq:pseudoexpectation1}
\begin{aligned}
F((x^2-x)m) &= F((x^2-x)(p^2 - q^2)) \\
& = F((x^2-x)p^2) + F((x-x^2)q^2) \\
& \geq 0, 
\end{aligned}
\end{equation}
\begin{equation}\label{eq:pseudoexpectation2}
\begin{aligned}
F((x^2-x)m) &= F((x^2-x)(p^2 - q^2)) \\
& = -F((x^2-x)q^2) - F((x-x^2)p^2) \\
& \leq 0,
\end{aligned}
\end{equation}
where the last inequalities in~(\ref{eq:pseudoexpectation1})
and~(\ref{eq:pseudoexpectation2}) follow from the fact that the
polynomials~$x^2-x$ and~$x-x^2$ are Boolean axioms so they belong
to~$\bar{Q}$ and~$d_{x^2-x} = d_{x-x^2} = 2d-2$.  Hence, by the
definition of a pseudoexpectation all the
values~$F((x^2-x)p^2)$,~$F((x-x^2)q^2)$,~$F((x^2-x)q^2)$
and~$F((x-x^2)p^2)$ are non-negative.

This shows that~$F((x^2-x)m) = 0$ and hence~$F(x^2m) = F(xm)$.

Now we show that~$0 \leq F(m) \leq 1$ for every monomial~$m$ of degree
at most~$d$. By the previous paragraph we have~$F(m) = F(m^2)$,
and~$F(m^2) \geq 0$ because~$m^2$ is a square of degree at
most~$2d$. The other inequality will be shown by induction on the
degree. For the empty monomial~$1$ we have~$F(1) = 1$. Now let~$m$ be
a monomial of degree at most~$d-1$ such that~$F(m) \leq 1$ and let~$x$
be a variable. It holds
that~$F(m) - F(xm) = F((1-x)m) = F((1-x)m^2) \geq 0$, and
hence~$F(xm) \leq F(m) \leq 1$.

Finally, let~$m$ be a monomial of degree at most~$2d$ and let~$r$
and~$s$ be monomials of degree at most~$d$ such that~$m = rs$. We
have~$F(r^2) + 2F(rs) + F(s^2) = F((r+s)^2) \geq 0$.
Therefore,~$2F(rs) \geq -F(r^2)-F(s^2) \geq -2$, so~$F(m) \geq
-1$. Similarly~$F(r^2) - 2F(rs) + F(s^2) = F((r-s)^2) \geq
0$. Therefore,~$2F(rs) \leq F(r^2) + F(s^2) \leq 2$, so~$F(m) \leq 1$.
\end{proof}

Let
\begin{equation}
\epsilon_{n,d} = \frac{1}{3} {n+2d \choose 2d}^{-1}. 
\end{equation}
Note that~$1/(3\epsilon_{n,d})$ is the number of monomials of degree~$2d$
over the set of~$n$ variables.  We are now ready to show that the
existence of a degree-$2d$~$\SOS$ refutation of a system of polynomial
inequalities with~$n$ variables is equivalent to the existence of
an~$\epsilon_{n,d}$-approximate such refutation.

\begin{proposition}\label{prop:relaxation}
  There exists an~$\SOS$ refutation of~$Q$ of degree~$2d$ if and only
  if there exists an~$\epsilon_{n,d}$-approximate~$\SOS$ refutation
  of~$Q$ of degree~$2d$, where~$n$ is the number of variables in~$Q$.
\end{proposition}

\begin{proof}
If~$Q$ has an~$\SOS$ refutation of degree~$2d$, then clearly it has
an~$\epsilon_{n,d}$-approximate refutation of degree~$2d$.

Now assume that~$Q$ has no~$\SOS$ refutation of
degree~$2d$. Therefore, by Corollary~\ref{cor:strongduality} there
exists a pseudoexpectation of degree~$2d$. Let us denote it
by~$F$. Suppose that~$Q$ has an~$\epsilon_{n,d}$-approximate~$\SOS$
refutation of degree~$2d$, i.e., there exists a set of sum of squares
polynomials~$\{ s_q : q \in \bar{Q} \}$ such that
\begin{equation}
\sum_{q \in \bar{Q}}qs_q = \sum_{\alpha}a_\alpha x^\alpha,
\end{equation}
where for each~$x^\alpha$ of degree at least~$1$, we have~$|a_\alpha|
\leq \epsilon_{n,d}$, and~$|1 + a_{\emptyset}| \leq \epsilon_{n,d}$.

Now, observe that~$F\bigl(\sum_{q \in \bar{Q}} q s_q\bigr) = \sum_{q
  \in \bar{Q}} F(q s_q) \geq 0$, while
\begin{equation}
F(\sum_{\alpha}a_\alpha x^\alpha) = a_{\emptyset} + \sum_{\alpha \neq \emptyset}a_\alpha F(x^\alpha) \leq -1 + \epsilon_{n,d} + {n+2d \choose 2d} \epsilon_{n,d}  \leq - {1 \over 3}.
\end{equation}
This contradiction finishes the proof.
\end{proof}

An~$\epsilon$-relaxation of a convex set~$K$ is either empty, which
clearly implies the emptiness of the set~$K$ itself, or it has volume
greater than~$\delta$ 
(see Lemma~\ref{lem:ball}),
where~$\delta$ can be easily computed by means
of~$\FPC$-interpretations from the representation of~$K$
and~$\epsilon$. We therefore get the
following:

\begin{corollary}
For every positive integer~$d$, there is an~$\FPC$-definable reduction
from the problem of deciding the existence of~$\SOS$ refutations of
degree~$2d$, to the weak feasibility problem for semidefinite sets.
\end{corollary}

\begin{proof}
The reduction is an~$\FPC$-interpretation which takes a set of
polynomials~$Q$ with~$n$ variables as input and outputs
the~$\epsilon_{n,d}$-relaxation of~$K_d(Q,-1)$ and a rational~$\delta
> 0$, such that either the~$\epsilon_{n,d}$-relaxation of~$K_d(Q,-1)$
is empty, or it has volume greater than~$\delta$.
\end{proof}

\section{Graph Isomorphism} \label{sec:isomorphism}

We formulate the isomorphism problem for graphs~$G$ and~$H$ as a
system~$\ISO(G,H)$ of quadratic polynomial equations with~$0/1$-valued
variables. Let~$U$ and~$V$ denote the sets of vertices of~$G$ and~$H$,
respectively, assumed to be disjoint. The atomic type of a tuple of points
in a relational structure is the complete description of the equalities and the
relations that the points in the tuple satisfy. In the special case of
graphs, these are the equalities and
the edge and non-edge relationships between the
vertices in the tuple. For~$u_1,u_2 \in U$, we write~$\tp_G(u_1,u_2)$
for the atomic type of~$(u_1,u_2)$ in~$G$.  Similarly,
for~$v_1,v_2 \in V$, we write~$\tp_H(v_1,v_2)$ for the atomic type
of~$(v_1,v_2)$ in~$H$. The system of equations has one~$0/1$-valued
variable~$x_{u,v}$ for each pair of vertices~$u \in U$ and~$v \in V$;
the intended meaning of~$x_{u,v} = 1$ is that the vertex~$u$ is mapped to~$v$
by a fixed isomorphism. The set of equations of~$\ISO(G,H)$ is the
following:
\begin{equation*}
\begin{array}{lll}
\sum_{v \in V} x_{u,v} - 1 = 0 & & \text{ for } u \in U, \\
\sum_{u \in U} x_{u,v} - 1 = 0 & & \text{ for } v \in V, \\
x_{u_1,v_1}x_{u_2,v_2} = 0 & & \text{ for } u_1,u_2 \in U, v_1,v_2 \in V \text{ s.t. }
\tp_G(u_1,u_2) \not= \tp_H(v_1,v_2).
\end{array}
\end{equation*}

When necessary, we think of the equations~$q = 0$ from~$\ISO(G,H)$ as
pairs of inequalities~$q \geq 0$ and~$-q \geq 0$.  It is
straightforward to check that the relational structure that
represents~$\ISO(G,H)$
can be produced from~$G$ and~$H$ by an~$\FPC$-interpretation.  As a
structure, the pair of graphs~$(G,H)$ is given by two sorts~$\bar{U}$
and~$\bar{V}$ for~$U$ and~$V$, and two binary relations~$E$ and~$F$ of
types~$\bar{U}\times\bar{U}$ and~$\bar{V} \times\bar{V}$ for the sets
of edges of~$G$ and~$H$, respectively. For sets of polynomial
equations and inequalities we use the representation described in
Section~\ref{sec:sos}.
 
\begin{fact} \label{fact:defiso} There is an~$\FPC$-interpretation
  that takes a pair of graphs~$(G,H)$ as input and outputs the set of
  equations~$\ISO(G,H)$.
\end{fact}

An SOS proof that~$G$ and~$H$ are not isomorphic is an SOS refutation
of~$\ISO(G,H)$. A Sherali-Adams ($\SA$) proof that~$G$ and~$H$ are not
isomorphic is an~$\SA$ proof of the inequality~$-1 \geq 0$
from~$\ISO(G,H)$, where an~$\SA$ proof is an identity of the
type~\eqref{eqn:sos} in which the polynomials~$s_q$ are not
sums-of-squares but sums of extended monomials, i.e., polynomials of
the
form~$\sum_{i \in I} c_i\prod_{j \in J_i} x_j \prod_{k \in
  K_i}(1-x_k)$ where each~$c_i$ is a positive real, and each~$J_i$
and~$K_i$ is a subset of indices of variables.
A Polynomial Calculus~($\PC$) proof that~$G$ and~$H$ are not
isomorphic is a~$\PC$ proof of the equation~$-1 = 0$ from the system
of polynomial equations~$\ISO(G,H)$, where by~$\PC$ we mean the
(deductive) proof system for deriving polynomial equations
over~$\reals[x_1,\ldots,x_n]$ by means of the following inference
rules:
from nothing derive the axiom polynomial equation~$x^2 - x = 0$, from
the equations~$p = 0$ and~$q = 0$ derive the equation~$p+q = 0$, and
from the equation~$p = 0$ derive the equations~$ap = 0$ and~$xp = 0$,
where~$p$ and~$q$ are polynomials,~$a$ is a real, and~$x$ is a
variable.  In monomial~$\PC$, as defined in~\cite{GroheBerkholz15},
the polynomial~$p$ in the last rule is required to be either a
monomial, or a product of a monomial with one of the polynomials
from the set of hypotheses (in our case~$\ISO(G,H)$), or a product
of a monomial and an axiom polynomial~$x^2-x$.

We rely on the following facts from~\cite{Atseriasdoi10}
and~\cite{GroheBerkholz15}:

\begin{theorem} \label{thm:halfAMandhalfBG}
Let~$G$ and~$H$ be graphs and let~$k$ be a positive integer.  The
following are equivalent:
\begin{enumerate} \itemsep=0pt
\item $G \equiv^{k} H$, i.e.,~$G$ and~$H$ cannot be distinguished
  by~$\Cinf{k}$-sentences,
\item there is no degree-$k$~$\SA$ proof that~$G$ and~$H$ are not
  isomorphic,
\item there is no degree-$k$ monomial~$\PC$ proof that~$G$ and~$H$ are
  not isomorphic.
\end{enumerate}
\end{theorem}

\noindent To be precise, the main result in \cite{Atseriasdoi10} is
stated for the formulation of the graph isomorphism problem as a
system of \emph{linear} equations with~$0/1$-valued variables.  For
that encoding, the correspondence between~$\equiv^k$-equivalence and
the non-existence of degree-$k$~$\SA$ proofs is not exact but only a
tight sandwich: if there is no degree-$k$~$\SA$ proof that~$G$ and~$H$
are not isomorphic then~$G \equiv^k H$, and if~$G \equiv^k H$ then
there is no degree-$(k-1)$~$\SA$ proof that~$G$ and~$H$ are not
isomorphic. However, it follows from the methods
in~\cite{Atseriasdoi10} and~\cite{GroheBerkholz15} that, for the
quadratic encoding used here, Theorem~\ref{thm:halfAMandhalfBG} holds
as stated.  For the collapse result we are about to prove, we use
Corollary~\ref{cor:decidingproofs} and the
implication~2. implies~1. from Theorem~\ref{thm:halfAMandhalfBG}.

\begin{theorem} \label{thm:collapse} There exists an integer~$c$ such
  that, for all pairs of graphs~$G$ and~$H$ and all positive
  integers~$d$, if there is a degree-$2d$~$\SOS$ proof that~$G$
  and~$H$ are not isomorphic, then there is a degree-$cd$~$\SA$ proof
  that~$G$ and~$H$ are not isomorphic.
\end{theorem}

\begin{proof}
  Fix a positive integer~$d$. Let~$\Phi$ be the FPC-interpretation
  from Fact~\ref{fact:defiso} and compose it with
  the~$\Cinf{\omega}$-sentence from Corollary~\ref{cor:decidingproofs}
  that decides the existence of SOS refutations of degree~$2d$. The
  resulting sentence~$\phi$ is a sentence of~$\Cinf{k}$, where~$k =
  cd$ for an integer~$c$ that is independent of~$d$.  The
  sentence~$\phi$ was designed in such a way that for every pair of
  graphs~$G$ and~$H$ it holds that~$(G,H) \models \phi$ if and only if
  there is a degree-$2d$ SOS proof that~$G$ and~$H$ are not
  isomorphic. In particular, since there certainly is no degree-$2d$
  SOS proof that~$G$ is not isomorphic to an isomorphic copy of
  itself, we have~$(G',G) \models \neg\phi$, where~$G'$ is an isomorphic
  copy of~$G$ on a disjoint set of vertices.
   Now assume that there is no
  degree-$k$~$\SA$ proof that~$G$ and~$H$ are not isomorphic. We
  get~$G \equiv^{k} H$ by Theorem~\ref{thm:halfAMandhalfBG}, from
  which it follows that~$(G',G) \equiv^{k} (G,H)$ because~$G' \cong G$
  and hence~$G' \equiv^k G$, and~$G \equiv^k H$.  Since~$\phi$ is
  a~$\Cinf{k}$-sentence and~$(G',G) \models \neg\phi$ we get~$(G,H)
  \models \neg\phi$. Therefore, by design of~$\phi$, there is no
  degree-$2d$ SOS proof that~$G$ and~$H$ are not isomorphic.
\end{proof}

Next we use the result of Berkholz
\cite{DBLP:journals/eccc/Berkholz17}
showing that, for systems of polynomial-equations over~$0/1$-valued
variables,~$\SOS$ simulates~$\PC$.

\begin{theorem}[\cite{DBLP:journals/eccc/Berkholz17}]\label{thm:SOSsimulatesPC}
  Let~$Q$ be a system of polynomial equations with real coefficients
  over~$0/1$-valued variables and let~$d$ be a positive
  integer. If~$Q$ has a~$\PC$ refutation of degree~$d$, then~$Q$ has
  an~$\SOS$ refutation of degree~$2d+1$.
\end{theorem}

The discrepancy between the~$2d+1$ in the conclusion of
Theorem~\ref{thm:SOSsimulatesPC} and the~$2d$ in the conclusion of
Theorem~1.1 from~\cite{DBLP:journals/eccc/Berkholz17} is due to a
small difference between our definition of~$\SOS$ and the variant
of~$\SOS$ used in~\cite{DBLP:journals/eccc/Berkholz17}. We discuss
this next. We focus on the case of polynomial equations, which is the
subject of the above theorem.  We denote the variant by~$\SOS'$.

Given a system of polynomial
equations~$Q = \{ q_i = 0 : i \in [k]\}$ over Boolean variables
and a
polynomial~$q$, an $\SOS'$ proof of~$q \geq 0$
from~$Q$
is a sequence of
polynomials~$(g_1,\ldots,g_k, h_1, \ldots, h_n, s_0)$ that satisfy
an identity
\begin{equation} \label{eqn:sosprime}
\sum_{i \in [k]}  q_i g_i + \sum_{j \in [n]}  (x^2_j - x_j) h_j + s_0 = q,
\end{equation}
where the polynomial~$s_0$ is a sum of squares of polynomials.  To be
able to compare~$\SOS$ with~$\SOS'$ we view each equation~$q_i = 0$
in~$Q$ as two inequalities~$q_i \geq 0$ and~$-q_i \geq 0$.  Hence,
an~$\SOS$ proof of~$q \geq 0$ from~$Q$ is a
set $\{s_1,\ldots,s_m\}$ of sum of squares polynomials that satisfy
the identity~$\sum_{j \in [m]} p_j s_j = q$, where, for
every~$j \in [m]$, the polynomial~$p_j$ either is in the
set~$\{ q_1, -q_1, \ldots, q_k, -q_k \}$ or is one of the Boolean
axioms listed in~\eqref{eq:axioms}.

\begin{lemma}
Let~$Q$ be a system of polynomial equations
  over~$0/1$-valued variables. If~$q \geq 0$ has an $\SOS'$
  proof from~$Q$ of degree $2d$, then it has an $\SOS$ proof
  from~$Q$ of degree at most $2d+1$.
\end{lemma}

\begin{proof}
  For any polynomial~$p$ and any monomial~$m$, such
  that~$\deg(pm) \leq 2d$, we will show that the product~$pm$ can be
  written as~$pm = p s + (-p)s'$, where~$s$ and~$s'$ are sums of
  squares of polynomials and~$\deg(ps) = \deg(-ps') \leq 2d + 1$. This
  last fact implies that the left-hand side of any degree-$2d$~$\SOS'$
  proof as in~\eqref{eqn:sosprime} can be rewritten as follows
\begin{equation}
\begin{aligned}
\sum_{i \in [k]}  q_i s_i + \sum_{i \in [k]}  (-q_i) s'_i + \sum_{j \in [n]}  (x^2_j - x_j) z_j + \sum_{j \in [n]}  (x_j - x^2_j) z'_j + s_0,
\end{aligned}
\end{equation}
where, in the
sequence~$(s_1, \ldots, s_k, s'_1, \ldots, s'_k, z_1, \ldots, z_n,
z'_1, \ldots, z'_n,s_0)$, all the polynomials are sums of squares.
As a result of this rewriting we obtain
an~$\SOS$ proof of~$q \geq 0$
from~$\{q_1,-q_1,\ldots,q_k,-q_k\}$, and
hence an~$\SOS$ proof of~$q \geq 0$
from~$Q$ by the convention that we introduced to
be able to compare~$\SOS$ with~$\SOS'$. Note that the degree of the
proof increases by at most~$1$.

Take any polynomial~$p$ and any monomial~$m$, such
that~$\deg(pm) \leq 2d$. Let~$r$ and~$t$ be
monomials such that $m = rt$ and $|\deg(r)-\deg(t)| \leq 1$.
Then we have $m = s - s'$, where $s = ((r+t)/2)^2$ and 
$s' = ((r-t)/2)^2$. Moreover, $\deg(s) = \deg(s') \leq \deg(m) +1$.
We obtain, $pm = p s + (-p)s'$,
where~$s$ and~$s'$ are sums of squares of polynomials and
$\deg(ps) = \deg(-ps') \leq 2d + 1$, which finishes the proof.
\end{proof}

For graphs~$G$ and~$H$,
let~$\mathrm{sos}(G,H)$,~$\mathrm{sa}(G,H)$,~$\mathrm{monpc}(G,H)$
and~$\mathrm{pc}(G,H)$ denote the smallest degrees for
which~$\SOS$,~$\SA$, monomial~$\PC$ and~$\PC$ prove that~$G$ and~$H$
are not isomorphic, respectively, taken as~$\infty$ if the graphs are
isomorphic.  Combining
Theorems~\ref{thm:SOSsimulatesPC}, \ref{thm:halfAMandhalfBG}, \ref{thm:collapse},
we get a full cycle of implications.

\begin{corollary}
  There exists an integer constant~$c$ such that, for all pairs of
  graphs~$G$ and~$H$, the following inequalities hold:
\begin{equation}
\tfrac{1}{2} \cdot (\mathrm{sos}(G,H)-1) \leq 
\mathrm{pc}(G,H) \leq 
\mathrm{monpc}(G,H) \leq 
\mathrm{sa}(G,H) \leq 
\tfrac{c}{2} \cdot \mathrm{sos}(G,H).
\end{equation}
\end{corollary}

Let us now state the collapse for the \emph{primals}. 
Recall from Subsection~\ref{subsec:lasserre} that
the Lasserre SDP relaxation of a polynomial optimization problem
is defined to be
a primal-dual pair of semidefinite
programs.
However, it is the primal that is most often referred to as the
Lasserre relaxation. This is the terminology we will use now.
For a positive integer~$k$, let~$\mathrm{LA}_k(G,H)$
denote the 
level-$k$ Lasserre relaxation of~$\ISO(G,H)$, i.e., the
primal in the primal-dual
SDP-pair~$(P_k(\ISO(G,H) ),D_k(\ISO(G,H) ))$ as defined in
Subsection~\ref{subsec:approx}. 
By Lemma~\ref{lem:refutation} and
the strong duality implied by Lemma~\ref{lem:linearcombination} and
Theorem~\ref{thm:strongduality}, for every positive integer~$d$ it holds
that~$\mathrm{LA}_{2d}(G,H)$ is feasible if and only if there is no
degree-$2d$~$\SOS$ proof that~$G$ and~$H$ are not isomorphic.
Similarly, we write~$\mathrm{SA}_k(G,H)$ to denote the primal in the
primal-dual LP-pair corresponding to the level-$k$ Sherali-Adams
relaxation of~$\ISO(G,H)$ as defined
in~\cite[Section~4]{SheraliAdams1990} for generic systems of
polynomial constraints over~$0/1$-valued variables. We refer to it as
the level-$k$ Sherali-Adams relaxation of~$\ISO(G,H)$. In this case,
strong duality holds by the duality theorem for linear programming,
and the dual solutions are degree-$k$~$\SA$ refutations
of~$\ISO(G,H)$. It follows that, for every positive~$d$, the linear
program~$\mathrm{SA}_d(G,H)$ is feasible if and only if there is no
degree-$d$~$\SA$ proof that~$G$ and~$H$ are not
isomorphic. Theorem~\ref{thm:collapse} gives then the following.

\begin{corollary} \label{cor:indirect} There exists an integer~$c$
  such that, for all pairs of graphs~$G$ and~$H$ and all positive
  integers~$d$, if the level-$2d$ Lasserre relaxation of~$\ISO(G,H)$
  is infeasible, then the level-$cd$ Sherali-Adams relaxation
  of~$\ISO(G,H)$ is infeasible.
\end{corollary}

As mentioned in the introduction, our proof of
Corollary~\ref{cor:indirect} is very indirect as it goes through many
black boxes. It would be very instructive to find a concrete and
direct way of lifting feasible LP-solutions of~$\mathrm{SA}_{cd}(G,H)$
to feasible SDP-solutions
of~$\mathrm{LA}_{2d}(G,H)$. Corollary~\ref{cor:indirect} and the fact
that its indirect proof is nonetheless constructive imply that such a
direct way of lifting does, in principle, exist.

\bigskip \noindent\textbf{Acknowledgments.} We are grateful to
Christoph Berkholz, Anuj Dawar, and Wied Pakusa, for useful
discussions at an early stage of this work.  We are also grateful to
Aaron Potechin for pointing out that the ability of the Lasserre
hierarchy to capture spectral arguments was relevant for our result.
Special thanks go to Moritz M\"uller for carefully reading and
commenting on a preliminary version of this paper.
First author partially funded by the
      European Research Council (ERC) under the European Union's
      Horizon 2020 research and innovation programme, grant agreement
      ERC-2014-CoG 648276 (AUTAR) and MICCIN grant TIN2016-76573-C2-1P
      (TASSAT3) and AEI grant PID2019-109137GB-C22 (PROOFS).  The work
      of second author on this manuscript is a part of the project
      BOBR that has received funding from the European Research
      Council (ERC) under the European Union's Horizon 2020 research
      and innovation programme (grant agreement No.~948057).  Second
      author supported also by the French Agence Nationale de la
      Recherche, QUID project reference ANR-18-CE40-0031. Part of this
      work was done while the second author was visiting UPC funded by
      AUTAR.

\bibliographystyle{plain}
\bibliography{bibfileforthis}

\begin{thebibliography}{10}

\bibitem{Allender:2009}
Eric Allender, Peter B\"{u}rgisser, Johan Kjeldgaard-Pedersen, and Peter~Bro
  Miltersen.
\newblock On the complexity of numerical analysis.
\newblock {\em SIAM J. Comput.}, 38(5):1987--2006, 2009.

\bibitem{Anderson:2015}
Matthew Anderson, Anuj Dawar, and Bjarki Holm.
\newblock Solving linear programs without breaking abstractions.
\newblock {\em J. ACM}, 62(6):48:1--48:26, 2015.

\bibitem{Atseriasdoi10}
Albert Atserias and Elitza Maneva.
\newblock Sherali--{Adams} relaxations and indistinguishability in counting
  logics.
\newblock {\em SIAM J. Comput.}, 42(1):112--137, 2013.

\bibitem{DBLP:journals/corr/abs-1802-02388}
Albert Atserias and Joanna Ochremiak.
\newblock Definable ellipsoid method, sums-of-squares proofs, and the
  isomorphism problem.
\newblock {\em CoRR}, abs/1802.02388, 2018.

\bibitem{10.1145/3209108.3209186}
Albert Atserias and Joanna Ochremiak.
\newblock Definable ellipsoid method, sums-of-squares proofs, and the
  isomorphism problem.
\newblock In {\em Proceedings of the 33rd Annual ACM/IEEE Symposium on Logic in
  Computer Science}, LICS '18, page 66–75, New York, NY, USA, 2018.
  Association for Computing Machinery.

\bibitem{DBLP:conf/stoc/BarakBHKSZ12}
Boaz Barak, Fernando G. S.~L. Brand{\~{a}}o, Aram~W. Harrow, Jonathan~A.
  Kelner, David Steurer, and Yuan Zhou.
\newblock Hypercontractivity, sum-of-squares proofs, and their applications.
\newblock In Howard~J. Karloff and Toniann Pitassi, editors, {\em Proceedings
  of the 44th Symposium on Theory of Computing Conference, {STOC} 2012, New
  York, NY, USA, May 19 - 22, 2012}, pages 307--326. {ACM}, 2012.

\bibitem{DBLP:journals/eccc/Berkholz17}
Christoph Berkholz.
\newblock {The Relation between Polynomial Calculus, Sherali-Adams, and
  Sum-of-Squares Proofs}.
\newblock In {\em Proceedings of the 35th Symposium on Theoretical Aspects of
  Computer Science}, pages 11:1--11:14, 2018.

\bibitem{GroheBerkholz15}
Christoph Berkholz and Martin Grohe.
\newblock Limitations of algebraic approaches to graph isomorphism testing.
\newblock In {\em Proceedings of the 42nd International Colloquium on Automata,
  Languages, and Programming, Part I}, pages 155--166, 2015.

\bibitem{Blass2002}
Andreas Blass, Yuri Gurevich, and Saharon Shelah.
\newblock On polynomial time computation over unordered structures.
\newblock {\em J. Symb. Log.}, 67(3):1093--1125, 2002.

\bibitem{DawarSeveriniZapata}
Anuj Dawar, Simone Severini, and Octavio Zapata.
\newblock Descriptive complexity of graph spectra.
\newblock {\em Ann. Pure Appl. Log.}, 170(9):993--1007, 2019.

\bibitem{DawarW17}
Anuj Dawar and Pengming Wang.
\newblock {Definability of semidefinite programming and Lasserre lower bounds
  for CSPs}.
\newblock In {\em Proceedings of the 32nd Annual ACM/IEEE Symposium on Logic in
  Computer Science}, pages 1--12, 2017.

\bibitem{DBLP:books/daglib/0082516}
Heinz{-}Dieter Ebbinghaus and J{\"{o}}rg Flum.
\newblock {\em Finite model theory}.
\newblock Perspectives in Mathematical Logic. Springer, 1995.

\bibitem{DBLP:books/daglib/0080659}
Heinz{-}Dieter Ebbinghaus, J{\"{o}}rg Flum, and Wolfgang Thomas.
\newblock {\em Mathematical logic {(2.} ed.)}.
\newblock Undergraduate texts in mathematics. Springer, 1994.

\bibitem{DBLP:journals/lmcs/GradelGPP19}
Erich Gr{\"{a}}del, Martin Grohe, Benedikt Pago, and Wied Pakusa.
\newblock A finite-model-theoretic view on propositional proof complexity.
\newblock {\em Log. Methods Comput. Sci.}, 15(1), 2019.

\bibitem{GrigorievVorobyov2001}
Dima Grigoriev and Nicolai Vorobjov.
\newblock Complexity of null- and positivstellensatz proofs.
\newblock {\em Annals of Pure and Applied Logic}, 113(1):153--160, 2001.
\newblock First St. Petersburg Conference on Days of Logic and Computability.

\bibitem{GroetschelLovaszSchrijverBook}
Martin Gr\"otschel, L\'aszl\'o Lov\'asz, and Alexander Schrijver.
\newblock {\em Geometric Algorithms and Combinatorial Optimization}.
\newblock Springer-Verlag, 1993.

\bibitem{DBLP:conf/lics/Hakoniemi21}
Tuomas Hakoniemi.
\newblock Monomial size vs. bit-complexity in sums-of-squares and polynomial
  calculus.
\newblock In {\em 36th Annual {ACM/IEEE} Symposium on Logic in Computer
  Science, {LICS} 2021, Rome, Italy, June 29 - July 2, 2021}, pages 1--7.
  {IEEE}, 2021.

\bibitem{10.1145/800070.802187}
Neil Immerman.
\newblock Relational queries computable in polynomial time (extended abstract).
\newblock In {\em Proceedings of the Fourteenth Annual ACM Symposium on Theory
  of Computing}, STOC '82, page 147–152, New York, NY, USA, 1982. Association
  for Computing Machinery.

\bibitem{Josz2016}
C{\'{e}}dric Josz and Didier Henrion.
\newblock Strong duality in {Lasserre's} hierarchy for polynomial optimization.
\newblock {\em Optim. Lett.}, 10(1):3--10, 2016.

\bibitem{Lasserre2001}
Jean{-}Bernard Lasserre.
\newblock Global optimization with polynomials and the problems of moments.
\newblock {\em SIAM J. Optim.}, 11(3):796--817, 2001.

\bibitem{MALKIN201473}
Peter~N. Malkin.
\newblock Sherali--{Adams} relaxations of graph isomorphism polytopes.
\newblock {\em Discrete Optim.}, 12:73--97, 2014.

\bibitem{DBLP:conf/innovations/ODonnell17}
Ryan O'Donnell.
\newblock {SOS} is not obviously automatizable, even approximately.
\newblock In Christos~H. Papadimitriou, editor, {\em 8th Innovations in
  Theoretical Computer Science Conference, {ITCS} 2017, January 9-11, 2017,
  Berkeley, CA, {USA}}, volume~67 of {\em LIPIcs}, pages 59:1--59:10. Schloss
  Dagstuhl - Leibniz-Zentrum f\"{u}r Informatik, 2017.

\bibitem{ODonnellSchramm2021}
Ryan O'Donnell and Tselil Schramm.
\newblock Sherali-adams strikes back.
\newblock {\em Theory Comput.}, 17:1--30, 2021.

\bibitem{O'Donnell:2014}
Ryan O'Donnell, John Wright, Chenggang Wu, and Yuan Zhou.
\newblock Hardness of robust graph isomorphism, {Lasserre} gaps, and asymmetry
  of random graphs.
\newblock In {\em Proceedings of the 25th Annual ACM-SIAM Symposium on Discrete
  Algorithms}, pages 1659--1677, 2014.

\bibitem{OttoBook}
Martin Otto.
\newblock {\em Bounded variable logics and counting -- {A} study in finite
  models}, volume~9.
\newblock Springer-Verlag, 1997.

\bibitem{Parrilo2000}
Pablo~A. Parrilo.
\newblock {\em Structured semidefinite programs and semialgebraic geometry
  methods in robustness and optimization}.
\newblock PhD thesis, Massachussets Institute of Technology, 2000.

\bibitem{DBLP:conf/icalp/RaghavendraW17}
Prasad Raghavendra and Benjamin Weitz.
\newblock On the bit complexity of sum-of-squares proofs.
\newblock In Ioannis Chatzigiannakis, Piotr Indyk, Fabian Kuhn, and Anca
  Muscholl, editors, {\em 44th International Colloquium on Automata, Languages,
  and Programming, {ICALP} 2017, July 10-14, 2017, Warsaw, Poland}, volume~80
  of {\em LIPIcs}, pages 80:1--80:13. Schloss Dagstuhl - Leibniz-Zentrum
  f\"{u}r Informatik, 2017.

\bibitem{10.5555/26851}
Walter Rudin.
\newblock {\em Real and Complex Analysis, 3rd Ed.}
\newblock McGraw-Hill, Inc., 1987.

\bibitem{SheraliAdams1990}
Hanif~D. Sherali and Warren~P. Adams.
\newblock A hierarchy of relaxations between the continuous and convex hull
  representations for zero-one programming problems.
\newblock {\em SIAM J. Discrete Math.}, 3(3):411--430, 1990.

\bibitem{TarasovVyalyi2008}
Sergey~P. Tarasov and Mikhail~N. Vyalyi.
\newblock Semidefinite programming and arithmetic circuit evaluation.
\newblock {\em Discrete Appl. Math.}, 156(11):2070--2078, 2008.

\bibitem{Tinhofer1986}
Gottfried Tinhofer.
\newblock Graph isomorphism and theorems of birkhoff type.
\newblock {\em Computing}, 36(4):285--300, 1986.

\bibitem{10.1145/800070.802186}
Moshe~Y. Vardi.
\newblock The complexity of relational query languages (extended abstract).
\newblock In {\em Proceedings of the Fourteenth Annual ACM Symposium on Theory
  of Computing}, STOC '82, page 137–146, New York, NY, USA, 1982. Association
  for Computing Machinery.

\bibitem{WangPhD}
Pengming Wang.
\newblock Descriptive complexity of constraint problems.
\newblock {\em Doctoral thesis}, 2018.

\end{thebibliography}

\end{document}